\documentclass[conference]{IEEEtran}
\IEEEoverridecommandlockouts 

\usepackage[utf8]{inputenc}   
\usepackage[T1]{fontenc}      
\usepackage{amsmath,amssymb,amsfonts} 
\usepackage{bm}               
\usepackage{dsfont}           
\usepackage{nicefrac}         
\usepackage{microtype,soul}   
\usepackage{xcolor}           
\usepackage{textcomp}         

\usepackage{graphicx}         
\graphicspath{{./figures/}}  
\usepackage{booktabs}         
\usepackage{multirow}         
\usepackage{tabularx}         
\usepackage[caption=false,font=normalsize,labelfont=sf,textfont=sf]{subfig}

\usepackage{algorithm}        
\usepackage{algorithmic}      

\usepackage{cite}             
\usepackage{url}              
\usepackage[colorlinks=true,citecolor=blue,breaklinks=true]{hyperref}

\newtheorem{theorem}{Theorem}[section]

\newtheorem{definition}[theorem]{Definition}
\newtheorem{lemma}[theorem]{Lemma}



\newcommand{\amax}{A_{\max}}
\renewcommand{\texttt}[1]{{\fontfamily{lmtt}\selectfont#1}}

\usepackage{refcount}         
\usepackage{tikz}             
\def\BibTeX{{\rm B\kern-.05em{\sc i\kern-.025em b}\kern-.08em
    T\kern-.1667em\lower.7ex\hbox{E}\kern-.125emX}}
\begin{document}

\title{Decentralized MARL for Coarse Correlated Equilibrium in Aggregative Markov Games\\
\thanks{This research is supported by the National Key Research and Development Program of China (2022YFA1004600), the Natural Science Foundation of China (T2293770, 12288201).}
}

\author{\IEEEauthorblockN{1\textsuperscript{st} Siying~Huang}
\IEEEauthorblockA{\textit{The School of Mathematical Sciences} \\
\textit{University of Chinese Academy of Sciences}\\
Beijing, China \\
huangsiying@amss.ac.cn}
\and
\IEEEauthorblockN{2\textsuperscript{nd} Yifen~Mu}
\IEEEauthorblockA{\textit{SKLMS} \\
\textit{AMSS, CAS}\\
Beijing, China \\
mu@amss.ac.cn}
\and
\IEEEauthorblockN{3\textsuperscript{rd} Ge~Chen}
\IEEEauthorblockA{\textit{SKLMS} \\
\textit{AMSS, CAS}\\
Beijing, China \\
chenge@amss.ac.cn}
}

\maketitle

\begin{abstract}
This paper studies the problem of decentralized learning of Coarse Correlated Equilibrium (CCE) in aggregative Markov games (AMGs), where each agent’s instantaneous reward depends only on its own action and an aggregate quantity. 
Existing CCE learning algorithms for general Markov games are not designed to leverage the aggregative structure, and research on decentralized CCE learning for AMGs remains limited. We propose an \textit{adaptive stage-based} V-learning algorithm that exploits the aggregative structure under a fully decentralized information setting. Based on the two-timescale idea, the algorithm partitions learning into stages and adjusts stage lengths based on the variability of aggregate signals, while using no-regret updates within each stage. We prove the algorithm achieves an $\epsilon$-approximate CCE in $\widetilde{O}(S A_{\max}T^5 /\epsilon^2)$ episodes, avoiding \textit{the curse of multiagents} which commonly arises in MARL. Numerical results verify the theoretical findings, and the decentralized, model-free design enables easy extension to large-scale multi-agent scenarios.
\end{abstract}

\begin{IEEEkeywords}
aggregative Markov games, coarse correlated equilibrium, decentralized learning, MARL, sample complexity
\end{IEEEkeywords}

\section{Introduction}
Many sequential decision-making problems in the real world involve strategic interactions among multiple agents in a shared environment. Multi-agent reinforcement learning (MARL) provides an effective framework for solving such problems, and has achieved success in many fields, including the game of Go \cite{silver2016mastering}, Poker \cite{brown2018superhuman}, autonomous driving \cite{kiran2021deep}, and Large Language Models \cite{zhu2025lamarl}. However, in practical systems, agents often make decisions based only on local observations, and centralized coordination is often impractical due to high communication costs or unreliable communication \cite{kar2013q,hodge2021deep}. Therefore, designing decentralized, communication-free, and scalable multi-agent learning algorithms is of great practical significance.

Markov games \cite{shapley1953stochastic} are a common mathematical framework for describing MARL. Under this framework, the core goal of multi-agent learning is usually to find a certain game equilibrium solution, such as Nash Equilibrium (NE) \cite{nash1950equilibrium} and its variants. However, in general-sum scenarios, computing NE is PPAD-hard \cite{daskalakis2009complexity,chen2009settling}, which makes learning NE face great challenges both in theory and practice.

A natural alternative is the coarse correlated equilibrium (CCE) \cite{aumann1987correlated}. Unlike NE, CCE in general-sum games can be computed in polynomial time \cite{papadimitriou2008computing}. More importantly, in normal-form games, CCE can be approximated when agents independently run no-regret learning algorithms \cite{hart2000simple}. This makes CCE well suited for decentralized learning. Accordingly, recent work has focused on learning CCE in Markov games and established several sample complexity guarantees \cite{liu2020sharp,mao2022provably,cai2024near}.

Despite these advances, most existing work treats the environment as fully general and does not exploit underlying structure \cite{jin2021v,mao2022improving}. In fact, in practical scenarios such as market competition \cite{iyer2014information} and public resource allocation \cite{mills1967aggregative}, an agent’s reward usually depends only on its own behavior and an aggregate quantity. Such problems with an aggregative reward structure can be modeled as aggregative Markov games (AMGs). This aggregative structure widely appears in demand response in power systems \cite{ye2016game}, resource allocation in communication networks \cite{liao2014efficient}, and pricing mechanisms based on aggregate demand in economic systems \cite{nocke2018multiproduct}, making AMGs a class of game models with important application relevance. 
However, existing work on AMGs is extremely limited. The few related studies mainly focus on network scenarios and continuous state spaces, and fail to consider decentralized CCE learning in AMGs. This leads to a natural question:
\begin{center}
\textbf{Can we design decentralized algorithms that exploit the aggregative structure to efficiently learn CCE?}
\end{center}

To address this question, this paper studies decentralized CCE learning in aggregative Markov games, and the main results are as follows:
\begin{itemize}
    \item We propose an adaptive stage-based V-learning algorithm. The algorithm follows a two-timescale idea. It partitions each episode into stages to create a stable learning environment, and adaptively adjusts stage lengths based on the variability of the aggregate signal. Within each stage, Tsallis-INF no-regret learning is used to quickly update the policy to achieve per-state no-regret. The algorithm is fully decentralized, and each agent makes decisions using only local information without communication.
    \item We show that if all agents independently run the proposed algorithm, an $\epsilon$-approximate CCE can be found in at most $\widetilde{O}(S A_{\max}T^5 /\epsilon^2)$ episodes, where $S$ is the number of states, $A_{\max}$ is the maximum action space size, and $T$ is the episode length. This upper bound avoids  dependence on the number of agents and matches existing sample complexity bounds of general Markov games \cite{jin2021v,mao2022improving}.
    \item We provide numerical results to support the theoretical findings. Since the algorithm is decentralized and model-free, it can be easily scaled to large multi-agent systems. 
\end{itemize}

\textbf{Paper Organization:} Section~\ref{sec:preliminaries} introduces preliminary definitions and notations; Section~\ref{sec:cce_ce} details the proposed decentralized adaptive stage-based learning algorithm and core theoretical results; Section~\ref{sec:simulations} presents numerical simulations and analysis; Section~\ref{sec:conclusions} concludes the paper.

\section{Preliminaries}\label{sec:preliminaries}

\noindent\textbf{Markov game.} An $N$-player episodic Markov game is defined by a tuple: 
\begin{equation}\label{MG1}
\{\mathcal{T},\mathcal{N},\!\{\mathcal{S}^t\}_{t=1}^T,\!\{\mathcal A_i^t(\cdot)\}_{i\in\mathcal N,t\in\mathcal T},\!\{r_i^t\}_{i\in\mathcal N,t\in\mathcal T},\!\{p^t\}_{t\in\mathcal T}\},
\end{equation}
where  i) $\mathcal{T} =\{1,2,\dots,T\}$ is the set of finite time steps in each episode; ii) $\mathcal{N} = \{1,2,\dots,N\}$ is the set of agents; iii) $\mathcal{S}^t$ is the finite state space at time $t$, where \( s^t \in \mathcal{S}^t \) denotes the state of the system at time \( t \), and $\{\mathcal{S}^t\}_{t=1}^T$ denote the collection of stage-wise state spaces; iv) $\mathcal A_i^t(s^t)$ denotes the finite action space available to agent $i$ at time $t$ when the system is in state $s^t\in\mathcal S^t$, where \( a_i^t \in \mathcal A_i^t(s^t) \) is the action taken by agent \( i \).
      The joint action space at time $t$ is given by $
      \mathcal A^t(s^t):=\prod_{i\in\mathcal N}\mathcal A_i^t(s^t)$, and \( \bm{a}^t = (a_1^t,\dots, a_N^t) \in \mathcal A^t(s^t) \) denotes the action profile at time \( t \);
    v) $r_i^t(s^t,\bm{a}^t) \in [0,1]$ is the stage payoff  (reward) of agent $i$ at time $t$; and vi) $p^t(\cdot\mid s^t,\bm{a}^t)\in\Delta(\mathcal S^{t+1})$ is the state transition kernel at time $t$, where $\Delta(\mathcal S^{t+1})$ denotes the probability simplex over $\mathcal S^{t+1}$.

Notably, this formulation \eqref{MG1} allows time-varying state spaces $\mathcal{S}^t$ and stage-dependent action spaces $\mathcal A_i^t(s^t)$, which captures finite-horizon problems with evolving constraints.

For notational convenience, we define
$S := \max_{t\in\mathcal T} |\mathcal S^t|$, $A_i := \max_{t\in\mathcal T}\max_{s^t\in\mathcal S^t} |\mathcal A_i^t(s^t)|, \forall i\in\mathcal N$,
and
$A_{\max} := \max_{i\in\mathcal N} A_i$.

\vspace{.8em}
\noindent\textbf{Aggregative Markov game.}
We consider a specialized class of Markov games with aggregative stage-wise reward structure, named \textit{aggregative Markov games} (AMGs). Formally, the stage payoff (reward) of each agent $i$ depends only on its local action and an aggregate of other agents' actions:
\[
r_i^t(s^t,\bm a^t) = r_i^t\bigl(s^t, a_i^t, \sigma(\bm a_{-i}^t)\bigr), \forall i \in \mathcal{N}, t\in \mathcal{T}, 
\]
where $\bm{a}_{-i}^t = (a_j^t)_{j\neq i}$ denotes the action profile of all agents except \(i\), and function $\sigma(\cdot)$ is an aggregator (e.g., sum, average).

Such structures arise in many large-scale systems where interactions are mediated through aggregate quantities, including congestion control, resource allocation in communication networks, and economic markets, where payoffs depend on summary statistics such as total demand or average load. This makes AMGs practically important and widely applicable, yet research on their decentralized learning remains scarce.

\vspace{.8em}
\noindent\textbf{Policy and value function.} A (Markov) policy $\bm{\pi}_i(\cdot):=(\pi_i^1(\cdot),\dots,\pi_i^T(\cdot))$ for agent $i$ is a sequence of step-$t$ decision rules such that $\pi_i^t(s^t)\in\Delta (\mathcal{A}_i^t(s^t))$ for all $t \in \mathcal{T}$ and $s^t \in \mathcal{S}^t$, where $\pi_i^t(\cdot)$ maps the step-$t$ state $s^t$ to a probability distribution over agent $i$'s feasible action space $\mathcal{A}_i^t(s^t)$ at state $s^t$.  Let $\Pi_i$ denote the set of all Markov policies for agent $i$, and $\Pi = \times_{i=1}^N \Pi_i$ denote the joint Markov policy space. Each agent aims to find a policy that maximizes its cumulative stage payoff over the $T$ stages. A joint policy (or policy profile) $\bm\pi = (\bm\pi_1,\dots,\bm\pi_N) \in \Pi$ induces a probability measure over the sequence of states and joint actions. 

For a policy profile $\bm\pi$, and for any $t\in\mathcal{T}$, $s \in \mathcal{S}^t$, and $\bm{a} \in \mathcal{A}^t(s^t)$, the value function and the state-action value function (or $Q$-function) for agent $i$ are defined as: 
\begin{align}
V_{t,i}^{\bm\pi} (s) &:= \mathbb{E}_{\bm\pi} \bigg[ \sum_{t'=t}^{T}r^{t'}_i(s^{t'},\bm{a}^{t'})\mid s^t = s \bigg],\label{eqn:value} \\
Q_{t,i}^{\bm\pi} (s,\bm{a}) &:= \mathbb{E}_{\bm\pi} \bigg[ \sum_{t'=t}^{T}r^{t'}_i (s^{t'},\bm{a}^{t'})\mid s^t = s, \bm{a}^t = \bm{a} \bigg]. \nonumber
\end{align}

\noindent\textbf{Best response and Nash equilibrium.}
 For agent $i$, a policy $\bm\pi_i^{\star}$ is a \emph{best response} to $\bm\pi_{-i}$ for a given initial state $s_1$ if $V_{1,i}^{\bm\pi_i^{\star}, \bm\pi_{-i}}(s_1) = \sup_{\bm\pi_i} V_{1,i}^{\bm\pi_{i }, \bm\pi_{-i}}(s_1)$. A policy profile $\bm\pi = (\bm\pi_{i}, \bm\pi_{-i })\in\Pi $ is a \emph{Nash equilibrium} (NE) if $\bm\pi_{i }$ is a best response to $\bm\pi_{-i}$ for all $i\in\mathcal{N}$.

\vspace{.8em}
\noindent\textbf{Correlated policy.}
More generally, we define a (non-Markov) \emph{correlated policy} as \( \bm\pi = (\pi^1(\cdot),\ldots,\pi^T(\cdot)) \), where for each time step \( t\in\mathcal{T} \), the step-\( t \) decision rule \( \pi^t \) satisfies: 
\( \pi^t(z, s^1,\bm{a}^{1},\dots,s^{t-1},\bm{a}^{t-1},s^{t}) \in \Delta\left(\mathcal{A}^t(s^t)\right) \) for all \( s^{t'} \in\mathcal{S}^{t'} \), \( \bm{a}^{t'}\in \mathcal{A}^{t'}(s^{t'}) \) (for \( t'\in\{1,\dots,t-1\} \)), \( s^{t} \in\mathcal{S}^{t} \), and \( z\in\mathcal{R} \). The decision rule \( \pi^t \) maps a random variable \( z\in\mathcal{R} \) and a history of length $t-1$ represented by \( (s^1,\bm{a}^1,\dots,s^{t-1},\bm{a}^{t-1}) \) to a distribution over the joint action space. We assume that the agents following a correlated policy can access a common source of randomness (e.g., a common random seed) for the random variable $z$. We let \( \bm\pi_i = (\pi_i^1(\cdot),\ldots,\pi_i^T(\cdot)) \) and \( \bm\pi_{-i} = (\pi_{-i}^1(\cdot),\ldots,\pi_{-i}^T(\cdot)) \) be the proper marginal policies of \( \bm\pi \), whose step-\( t \) outputs are restricted to \( \Delta(\mathcal{A}_i^t(s^t)) \) and \( \Delta(\mathcal{A}_{-i}^t(s^t)) \), respectively. 

For non-Markov correlated policies, the value function at $t=1$ is defined analogously to~\eqref{eqn:value}. A best response $\bm\pi_i^{\star}$ with respect to the non-Markov policies $\bm\pi_{-i}$ is a policy (independent of the randomness of  $\bm\pi_{-i}$) that maximizes agent $i$'s value at step 1, i.e., $V_{1,i}^{\bm\pi_i^{ \star}, \bm\pi_{-i}}(s_1) = \sup_{\bm\pi_i} V_{1,i}^{\bm\pi_{i }, \bm\pi_{-i}}(s_1)$, and is not necessarily Markov. 

\vspace{.8em}
\noindent\textbf{Coarse correlated equilibrium.}
Given the PPAD-hardness of calculating NE in general-sum games~\cite{daskalakis2009complexity}, a widely used relaxation is \textit{coarse correlated equilibrium} (CCE).  A CCE ensures no agent has the incentive to deviate from a correlated policy $\bm\pi$ by playing a different independent policy. 

\begin{definition}\label{def:CCE}
	(CCE). A correlated policy $\bm\pi$ is an $\epsilon$-approximate coarse correlated equilibrium for any initial state $s^1\in \mathcal{S}^1$ if
	$
	V_{1,i}^{\bm\pi_i^{ \star}, \bm\pi_{-i}}(s^1) - V_{1,i}^{\bm\pi}(s^1) \leq \epsilon,\forall i\in\mathcal{N}.
	$
\end{definition}

CCE relaxes NE by allowing possible correlations in the policies, and NE is a special CCE in general-sum games~\cite{nisan2007algorithmic}. 

\begin{algorithm*}[!tbp]
\caption{Adaptive Stage-Based V-Learning with Tsallis-INF for CCE (Agent $i$)}\label{alg:sbv}
\begin{algorithmic}[1] 
    \STATE \textbf{Initialize:} 
    $\overline{V}_{t,i}(s) \gets T-t+1$, $\hat{V}_{t,i}(s) \gets T-t+1$, $\mathcal{D}^t(s) \gets \emptyset$, $\tilde{C}^t(s)\gets 0$, 
    $\tilde{r}_{i}^{t}(s)\gets 0$, $\tilde{v}_{i}^{t}(s)\gets 0$, \\ $\tilde{L}^t(s)\gets T$, $\pi_{i}^{t}(a\mid s)\gets 1/|\mathcal A_i^t(s)|$, 
    $\hat Q_{t,i}(s,a)\gets 0$, $\forall t\in\mathcal{T},s\in\mathcal{S}^t,a\in\mathcal A_i^t(s)$.
    
    \FOR{episode $k\gets 1$ to $K$}
        \STATE Receive initial state $s^1$;
        \FOR{step $t\gets 1$ to $T$}
            \STATE $\tilde{c}:= \tilde{C}^t(s^t) \gets \tilde{C}^t(s^t) + 1$;
            \STATE Take action $a_{i}^{t} \sim \pi_{i}^{t}(\cdot \mid s^t)$, observe reward $r_{i}^{t}$ and next state $s^{t+1}$, and compute aggregate quantity $d^t$\label{line:6};
            \STATE $\mathcal{D}^t(s^t) \gets \mathcal{D}^t(s^t) \cup \{d^t\}$\label{line:7};
            \STATE $\tilde{r}_{i}^{t}(s^t)\gets \tilde{r}_{i}^{t}(s^t) + {r}_{i}^{t}$, $\tilde{v}_{i}^{t}(s^t)\gets \tilde{v}_{i}^{t}(s^t) + \overline{V}_{t+1,i}(s^{t+1})$;
            \STATE $\eta_i \gets 2\sqrt{1/\tilde{c}}$;
            \STATE $\hat Q_{t,i}(s^t,a_{i}^{t}) \gets \hat  Q_{t,i}(s^t,a_{i}^{t})+ \frac{[T-t+1-({r}_{i}^{t} + \overline{V}_{t+1,i}(s^{t+1}))]/T}{\pi_{i}^{t}(a_{i}^{t} \mid s^t)}$\label{line:10};
            \STATE $\pi_{i}^{t}(a\mid s^t)\gets 4(\eta_i (\hat Q_{t,i}(s^t,a)-x))^{-2} ,\forall a\in\mathcal A_i^t(s^t),$ where normalization factor $x$ satisfies\\ $\sum_{a\in\mathcal A_i^t(s^t)} 4(\eta_i (\hat Q_{t,i}(s^t,a)-x))^{-2} =1$\label{line:11};
            
            \IF{$\tilde{C}^t(s^t) = \tilde{L}^t(s^t)$ \label{line:12}}
                \STATE \texttt{ // Entering a new stage}
                \STATE $\hat{V}_{t,i}(s^t)\gets  \frac{\tilde{r}_{i}^{t}(s^t)}{\tilde{c}} + \frac{\tilde{v}_{i}^{t}(s^t)}{\tilde{c}}+b_{\tilde{c}}$, where $b_{\tilde{c}}\gets 4\sqrt{T^2 A_i\iota / \tilde{c}}$\label{line:14};
                \STATE $\overline{V}_{t,i}(s^t)\gets \min\{\hat{V}_{t,i}(s^t),T-t+1\}$;
                \STATE $\lambda(s^t) \gets f(\mathcal{D}^t(s^t))$ \texttt{//Adaptive fluctuation coefficient (Alg. \ref{alg:fluctCV}/\ref{alg:fluctMAD})}\label{line:16}; 
                \STATE $\tilde{L}^t(s^t)\gets \lfloor \lambda(s^t) (1+\frac{1}{T})\tilde{L}^t(s^t) \rfloor$\label{line:17};
                \STATE $\mathcal{D}^t(s^t) \gets \emptyset$, $\tilde{C}^t(s^t)\gets 0$, $\tilde{r}_{i}^{t}(s^t)\gets 0$, $\tilde{v}_{i}^{t}(s^t)\gets 0$;
                \STATE $\pi_{i}^{t}(a\mid s^t)\gets 1/|\mathcal A_i^t(s^t)|$; $\hat Q_{t,i}(s^t,a)\gets 0$, $\forall a\in\mathcal A_i^t(s^t)$\label{line:19};
            \ENDIF
        \ENDFOR
    \ENDFOR
\end{algorithmic}
\end{algorithm*}

\vspace{.8em}
\noindent\textbf{Decentralized learning.} Agents interact in an unknown environment for $K$ episodes, with the initial state $s^1$ drawn from a fixed distribution $\rho\in\Delta(\mathcal{S}^1)$. At each step $t\in\mathcal{T}$, the agents observe the state $s^t \in \mathcal{S}^t$, and take actions $a^t_i \in\mathcal A_i^t(s^t), i\in\mathcal{N}$ simultaneously.  Agent $i$ then receives its private stage payoff $r_{i}^t(s^t,\bm{a}^t)$, and the environment transitions to the next state $s^{t+1}\sim p^t(\cdot | s^t,\bm{a}^t)$. Notably, non-deterministic state transitions raise the difficulty of decentralized coordination, as agents cannot rely on fixed state trajectories for implicit alignment.

We focus on a fully \emph{decentralized} setting: each agent only observes the states and its own rewards and actions, but not the rewards or actions of the other agents. In fact, in our algorithm, each agent is completely oblivious to the existence of the others, and does not communicate with each other. This decentralized information structure requires decision-making based solely on local information, and naturally arises in practical multi-agent systems where communication is limited or costly.

\section{Decentralized CCE Learning in AMGs}\label{sec:cce_ce}

This section presents our decentralized learning algorithm for CCE in general-sum AMGs and establishes its sample complexity guarantees.

\subsection{Algorithmic Design Ideas}
In decentralized Markov games, from the perspective of any individual agent, the environment reduces to a single-agent Markov decision process (MDP) within a single episode. Across multiple episodes, however, the environment perceived by each agent becomes non-stationary, making it highly challenging for agents to simultaneously estimate value functions and optimize policies.

A natural way to address this non-stationarity is a \textbf{two-timescale} design, a validated approach in decentralized multi-agent learning~\cite{sayin2021decentralized,sayin2022fictitious}: 
value estimation is updated on a slower timescale to maintain a relatively stable learning target, while policy updates run on a faster timescale within stable phases.
To implement this idea, we partition the learning process into \emph{stages} with slowly evolving value functions, inspired by \emph{stage-based} V-learning~\cite{mao2022improving}. Within each stage, we optimize policies via \emph{Tsallis-INF} to achieve \emph{per-state no-regret} learning.
This combination of staged value stabilization and no-regret policy updates forms the core of our decentralized CCE learning framework.

\subsection{Algorithm Description}
Compared to general Markov games, the \textbf{aggregative payoff} structure reduces effective interaction complexity and provides additional information through observable aggregate quantities that summarize collective behavior.

Our algorithm leverages this structure via an \textbf{adaptive stage-based} mechanism, as illustrated in Figure~\ref{fig:adaptive_stage}. Agents first compute aggregate signals within each stage, then estimate a fluctuation coefficient to assess environmental stability, and finally adapt stage lengths based on this stability signal. This design allows the algorithm to adjust to varying levels of non-stationarity, distinguishing it from existing fixed-stage methods.

\begin{figure}[tbp]
    \centering
    \includegraphics[width=0.49\textwidth]{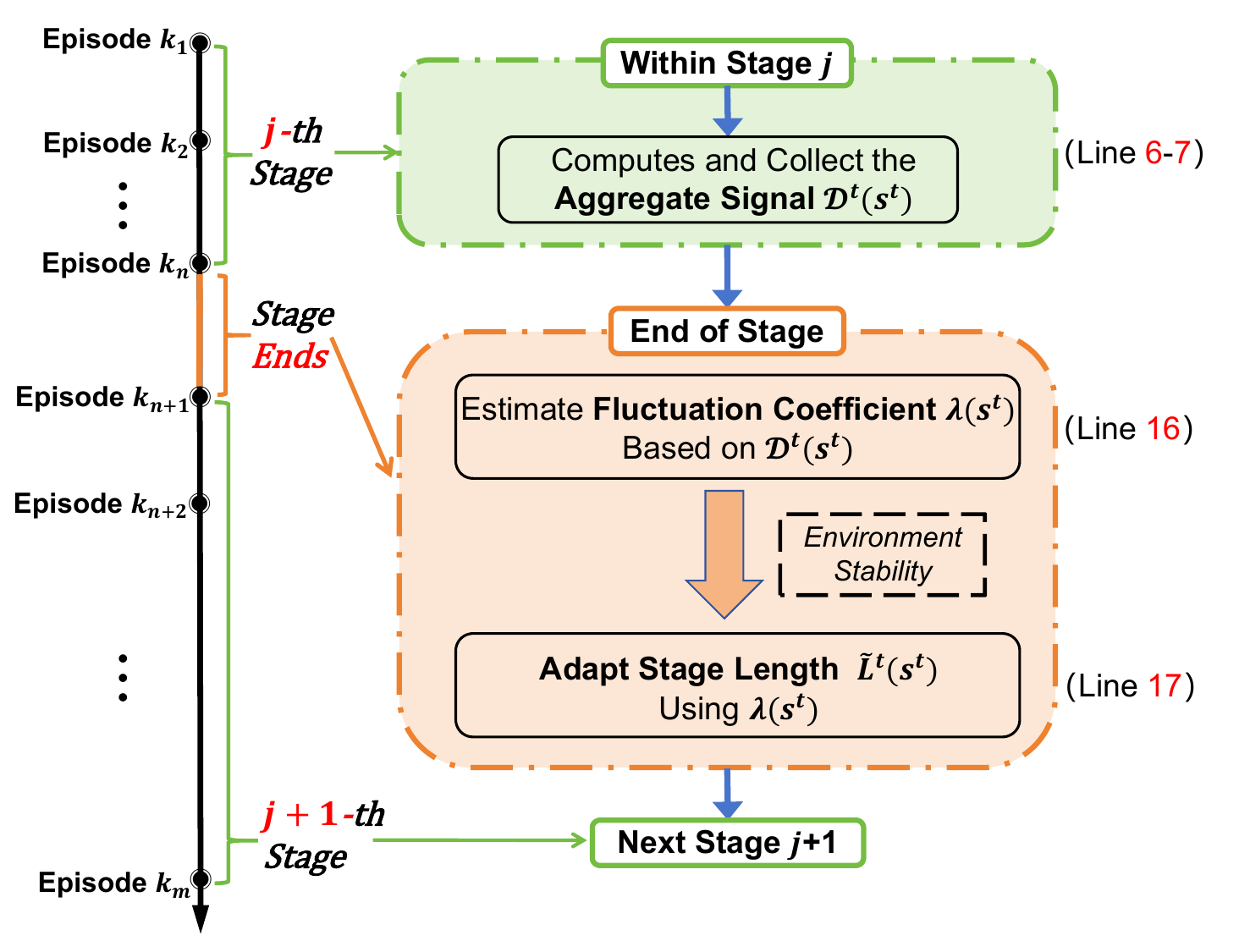}
    \caption{Illustration of the adaptive stage-based mechanism.}
    \label{fig:adaptive_stage}
\end{figure}

The adaptive stage-based V-learning algorithm for CCE is presented in Algorithm~\ref{alg:sbv}, which is executed independently by each agent $i\in\mathcal{N}$. The agent maintains upper confidence bounds on the value functions to actively explore the unknown environment, and updates value estimates independently via the aforementioned adaptive stage-based rule.

For each step-state pair $(t,s^t)$, the agent partitions its visitations into a sequence of \emph{stages}. Each stage has a length $\tilde{L}^t(s^t)$, initialized as $\tilde{L}^t(s^t)=T$ and updated at the end of each stage. The stage length is adjusted via a \emph{fluctuation coefficient} $\lambda(s^t)$ computed from the sequence of observed aggregate quantities within the current stage, denoted by $\mathcal{D}^t(s^t)=\{d^{t,1},\dots,d^{t,n}\}$. Specifically, the next stage length is updated as
\[
\tilde{L}^t(s^t)\gets \lfloor \lambda(s^t)(1+\tfrac{1}{T})\tilde{L}^t(s^t)\rfloor,
\]
so that stage lengths adapt to the observed aggregate volatility of the environment while growing at a near-geometric rate $(1+1/T)$ \cite{zhang2020almost}. In practice, $\lambda(s^t)$ can be computed using various stability metrics. We adopt two common approaches: Coefficient of Variation (CV) \cite{hendricks1936sampling} (see  Algorithm~\ref{alg:fluctCV}) or Mean Absolute Deviation (MAD) \cite{geary1935ratio} (see Algorithm~\ref{alg:fluctMAD}).

\begin{algorithm}[!tbp]
\caption{Adaptive Fluctuation Coefficient $f_{\text{CV}}(\mathcal{D}^t(s^t))$ via CV}\label{alg:fluctCV}
\begin{algorithmic}[1]
    \STATE \textbf{Input:} Positive aggregate quantity sequence $\mathcal D^t(s^t)=\{d^{t,1},\dots,d^{t,n}\}$, minimum fluctuation coefficient $\lambda_{\text{min}} \in (\frac{T}{T+1},1]$, empirical upper bound $CV_{\text{max}}$
    \STATE \textbf{Output:} Adaptive fluctuation coefficient $\lambda(s^t)\in[\lambda_{\min},1]$

    \IF{$|\mathcal{D}^t(s^t)| < 2$}
        \STATE $\lambda(s^t) \gets 1$; \texttt{// Default: no fluctuation (max stage length)}\;
        \RETURN $\lambda(s^t)$
    \ENDIF

    \STATE $\bar{d} \gets \frac{1}{|\mathcal{D}^t(s^t)|}\sum_{d^{t,j} \in \mathcal{D}^t(s^t)} d^{t,j}$;

    \STATE $\sigma \gets \sqrt{\frac{1}{|\mathcal{D}^t(s^t)|-1}\sum_{d^{t,j} \in \mathcal{D}^t(s^t)} (d^{t,j} - \bar{d})^2}$; 
    \STATE $CV \gets \frac{\sigma}{\bar{d}}$;

    \STATE $\gamma \gets \min\left(\frac{CV}{CV_{\text{max}}}, 1\right)$; 
    \STATE $\lambda(s^t) \gets \lambda_{\text{min}} + (1 - \lambda_{\text{min}}) \cdot \gamma$;
    \RETURN $\lambda(s^t)$  
\end{algorithmic}
\end{algorithm}

\begin{algorithm}[!tbp]
\caption{Adaptive Fluctuation Coefficient $f_{\text{MAD}}(\mathcal{D}^t(s^t))$ via MAD}\label{alg:fluctMAD}
\begin{algorithmic}[1]
    \STATE \textbf{Input:} Positive aggregate quantity sequence $\mathcal D^t(s^t)=\{d^{t,1},\dots,d^{t,n}\}$, minimum fluctuation coefficient $\lambda_{\text{min}} \in (\frac{T}{T+1},1]$, empirical upper bound $MAD_{\text{max}}$
    \STATE \textbf{Output:} Adaptive fluctuation coefficient $\lambda(s^t) \in [\lambda_{\text{min}}, 1]$

    \IF{$|\mathcal{D}^t(s^t)| < 2$}
        \STATE $\lambda(s^t) \gets 1$; \texttt{// Default: no fluctuation}\;
        \RETURN $\lambda(s^t)$
    \ENDIF

    \STATE $\bar{d} \gets \frac{1}{|\mathcal{D}^t(s^t)|}\sum_{d^{t,j} \in \mathcal{D}^t(s^t)} d^{t,j}$;

    \STATE $MAD \gets \frac{1}{|\mathcal{D}^t(s^t)|}\sum_{d^{t,j} \in \mathcal{D}^t(s^t)} |d^{t,j} - \bar{d}|$;

    \STATE $\gamma \gets \min\left(\frac{MAD}{MAD_{\text{max}}}, 1\right)$; 
    \STATE $\lambda(s^t) \gets \lambda_{\text{min}} + (1 - \lambda_{\text{min}}) \cdot \gamma$;
    \RETURN $\lambda(s^t)$
\end{algorithmic}
\end{algorithm}

When the visitation count $\tilde{C}^t(s^t)$ reaches $\tilde{L}^t(s^t)$, the current stage ends and the agent updates its optimistic value estimate $\overline{V}_{t,i}(s^t)$ using only samples collected within this stage (Lines~\ref{line:12}--\ref{line:19}). All stage-specific statistics are then reset, and the policy is reinitialized to a uniform distribution for the next stage. This staged update helps maintain a relatively stable learning environment, mitigate multi-agent non-stationarity, and its near-geometric stage length growth aligns with the update logic of optimistic Q-learning with learning rate $\alpha^t = \frac{T+1}{T+t}$ \cite{jin2018q,bai2019provably}.

At each time step $t$ and state $s^t$, agent $i$ selects its action $a_{i}^{t}$ by following a distribution $\pi_{i}^{t}(\cdot \mid s^t)$, which is updated via the Tsallis-INF adversarial bandit subroutine~\cite{zimmert2021tsallis} (Lines~\ref{line:10}--\ref{line:11}) to guarantee per-state no-regret learning. The normalization factor in the policy update can be computed efficiently using Newton's method, with detailed steps provided in Algorithm~\ref{alg:Newton}.

To obtain the final $\epsilon$-approximate CCE policy, we construct a unified output policy $\bar{\bm{\pi}}$ following the certified policy framework \cite{bai2020near}, as detailed in Algorithm~\ref{alg:certify}. Let $\pi_{i}^{t,k}(\cdot \mid s^{t})$ denote the policy of agent $i$ at step $t$ of episode $k$ under state $s^{t}$ generated by Algorithm~\ref{alg:sbv}. Agents use a shared random seed to uniformly sample an episode index from the previous stage, yielding the final correlated policy. 
This mild common randomness serves as a standard correlation device, which is only used for post-learning policy synchronization and does not break the fully decentralized learning paradigm.

\begin{algorithm}[!tbp]
\caption{Newton’s Method approximation of $\pi_{i}^{t}(a\mid s^t)$ in Tsallis-Inf (Line~\ref{line:11} in Algorithm~\ref{alg:sbv})}\label{alg:Newton}
\begin{algorithmic}[1]
    \STATE \textbf{Input:} $s^t, x, \hat Q_{t,i}(\cdot),\eta_i$ \texttt{//we use x from the last iteration as a warmstart}\;

    \REPEAT
        \STATE $\forall a\in\mathcal{A}_i^t(s^t): \pi_{i}^{t}(a\mid s^t)\gets 4(\eta_i (\hat Q_{t,i}(s^t,a)-x))^{-2}$;
        \STATE $x\gets x-\frac{(\sum_{a\in\mathcal{A}_i^t(s^t)}\hat Q_{t,i}(s^t,a)-1)}{(\eta_i \sum_{a\in\mathcal{A}_i^t(s^t)}\hat Q_{t,i}(s^t,a)^{\frac{3}{2}})}$;
    \UNTIL{convergence}

\end{algorithmic}
\end{algorithm}

\begin{algorithm}[!tbp]
\caption{Construction of the Output Policy $\bar{\bm{\pi}}$}\label{alg:certify}
\begin{algorithmic}[1]
    \STATE \textbf{Input:} The distribution trajectory specified by Algorithm~\ref{alg:sbv}: $\{\pi_{i}^{t,k}:i\in\mathcal{N},t\in \mathcal{T},k\in[K]\}$;

    \STATE Uniformly sample $k$ from $[K]$;
    \FOR{step $t\gets 1$ to $T$}
        \STATE Receive $s^{t}$;
        \STATE Uniformly sample $j$ from $\{1,2,\dots, \tilde{c}^{t,k}(s^{t})\}$; \COMMENT{For a state $s^t$, $\tilde{c}^{t,k}$ denotes the number of visits to the state $s^t$ (at the $t$-th step)  in the stage right before the current stage}\;
        \STATE Set $k \gets \tilde{l}_{j}^{t,k}$; \COMMENT{ $\tilde{l}_{j}^{t,k}$ is the index of the episode that this state was visited the $j$-th time among the total $\tilde{c}^{t,k}$ times} \;
        \STATE Take joint action $\bm{a}^{t}\sim \times_{i=1}^N\pi_{i}^{t,k}(\cdot \mid s^{t})$;
    \ENDFOR
\end{algorithmic}
\end{algorithm}

\subsection{Theoretical Guarantees}
The following theorem presents the sample complexity guarantee of Algorithm~\ref{alg:sbv} for learning CCE in general-sum aggregative Markov games. Our bound matches those established for general-sum Markov games in prior work~\cite{song2021can,jin2021v,mao2022improving}. 

\begin{theorem}\label{thm:cce}
	(Sample complexity of learning CCE). For any $p\in (0,1]$, set $\iota = \log(2NSA_{\max} KT/p)$, and let the agents run Algorithm~\ref{alg:sbv} for $K$ episodes with $K= O(S A_{\max}T^5 \iota/\epsilon^2)$. Then, with probability at least $1-p$, the output policy $\bar{\pi}$ of Algorithm~\ref{alg:certify} is an $\epsilon$-approximate CCE. 
\end{theorem}
\textit{Proof sketch:}
We provide a high-level overview of the proof and defer all technical details to Appendix~\ref{app:cce}.

The analysis proceeds in four main steps.

\textit{Step 1 (Estimation and regret control).}
We first establish a high-probability bound on the stage-wise value estimation error. By combining martingale concentration arguments with the no-regret guarantee of the Tsallis-INF algorithm, we show that the empirical value estimates concentrate around their expectations up to an error of order $O(\sqrt{T^2 A_i \iota / \tilde{c}})$, which justifies the choice of the exploration bonus used in the algorithm (Line~\ref{line:14}).

\textit{Step 2 (Confidence bounds).}
We construct optimistic upper value estimates $\overline{V}$ and pessimistic lower value estimates $\underline{V}$, and prove that they form valid high-probability confidence bounds. In particular, with high probability,
\[
V_{t,i}^{\star,\bar{\pi}_{-i}^{t,k}}(s^t) \leq \overline{V}_{t,i}^k(s^t), 
\quad
V_{t,i}^{\bar{\bm\pi}^{t,k}}(s^t) \geq \underline{V}_{t,i}^k(s^t),
\]
reducing the problem of bounding the CCE equilibrium gap to bounding the difference between these two estimates, $\overline{V}-\underline{V}$.

\textit{Step 3 (Recursive error propagation).}
Define $\delta^{t,k} := \overline{V}_{t,i}^k(s^{t,k}) - \underline{V}_{t,i}^k(s^{t,k})$. We show that $\delta^{t,k}$ satisfies a recursive inequality:
\[
\delta^{t,k} \leq \mathbb{I}[c^{t,k}=0]T + \frac{1}{\tilde{c}}\sum_{j=1}^{\tilde{c}} \delta^{t+1,\tilde{l}_j} + O(b_{\tilde{c}}).
\]
Unrolling this recursion and summing over episodes yields
\[
\sum_{k=1}^K \delta^{1,k} \leq O\big(ST^3 + \sqrt{SA_{\max}KT^5\iota}\big).
\]

\textit{Step 4 (Conclusion).}
Averaging over episodes and using the definition of the output policy $\bar{\bm\pi}$, we obtain
\[
V_{1,i}^{\star,\bar{\bm\pi}_{-i}}(s^1) - V_{1,i}^{\bar{\bm\pi}}(s^1)
\leq O\left(\sqrt{SA_{\max}T^5\iota/K}\right),
\]
which establishes the claimed sample complexity result.

\textit{Key Lemmas:}
The proof relies on the following intermediate results, whose detailed proofs are provided in Appendix~\ref{app:cce}.

\begin{lemma}[Estimation and regret bound]
	With probability at least $1-\frac{p}{2}$, it holds for all agents $i \in \mathcal{N},$ episodes $k \in [K],$  steps $t \in \mathcal{T},$  and states $ s^t \in \mathcal{S}^t$ that
\begin{align}
	&\max_{\pi_{i}^t} \frac{1}{\tilde{c}} \sum_{j=1}^{\tilde{c}} \mathbb{D}_{\pi_{i}^t\times \pi_{-i}^{t,\tilde{l}_j}}\left(r_{i}^t + \mathbb{P}_t\overline{V}_{t+1,i}^{\tilde{l}_j}\right)(s^t)\notag \\ &- \frac{1}{\tilde{c}} \sum_{j=1}^{\tilde{c}} \left(r_{i}^t(s^t,\bm{a}^{t,\tilde{l}_j}) + \overline{V}_{t+1,i}^{\tilde{l}_j}(s^{t+1,\tilde{l}_j}) \right) \leq 4\sqrt{T^2 A_i\iota/\tilde{c}}, \notag
    \end{align}
	 where $\iota = \log(4NS\amax KT/p)$.
\end{lemma}

\begin{lemma}[Confidence bounds]
It holds with probability at least $1-p$ that for all  for all agents $i \in \mathcal{N},$ episodes $k \in [K],$  steps $t \in \mathcal{T},$  and states $ s^t \in \mathcal{S}^t$ ,
\[
V_{t,i}^{\star,\bar{\pi}_{-i}^{t,k}}(s^t) \leq \overline{V}_{t,i}^k(s^t), 
\quad
V_{t,i}^{\bar{\bm\pi}^{t,k}}(s^t) \geq \underline{V}_{t,i}^k(s^t).
\]
\end{lemma}

\section{Simulations}\label{sec:simulations}
We empirically evaluate Algorithm~\ref{alg:sbv} on the classic Fishermen Game. 

The Fishermen Game is a simple $2$-horizon, $2$-agent Markov game adapted from \cite{maschler2020game}. It has two states $\mathcal{S}^t=\{s_h, s_\ell\}$ for $t=1,2$, where $s_h$ and $s_\ell$ denote high and low fish stock levels, respectively. Each fisherman has a binary action set $\mathcal{A}_i^t = \{a_m, a_f\}$ for $i = 1,2$ and $t=1,2$, where \( a_m = 5 \) represents intensive fishing with many nets, and \( a_f = 3 \) means mild fishing with few nets.

We define the \emph{aggregate action} as the total fishing effort across both agents, given by $a = a_1 + a_2$ where $a_i \in \mathcal{A}_i^t$.

The reward function for agent \( i \) in state \( s \) is formulated as an aggregative payoff (preserving the original game's reward values and structure):
\[
r_i(s, \bm a)=r_i(s, a_1, a_2)= r_i(s, a_i, a) = f(a_i) - g(a) - c(s)
\]
where  the concave private payoff function is  \( f(a_i) = -\frac{1}{2}a_i^2 + \frac{23}{2}a_i + 1 \), the convex aggregate cost function is \( g(a) = \frac{1}{4}a^2 - \frac{1}{2}a + 16 \), and the state-dependent cost is \( c(s_h) = 0 \) (no extra cost in high stock state), \( c(s_\ell) = 1 \) (extra cost in low stock state). The resulting payoff matrices (rows: agent 1, columns: agent 2) are:
\[
\begin{array}{c|cc}
s_h & a_m\ (5) & a_f\ (3) \\
\hline
a_m\ (5) & (10, 10) & (18, 3) \\
a_f\ (3) & (3, 18) & (9, 9)
\end{array}
\quad\quad
\begin{array}{c|cc}
s_\ell & a_m\ (5) & a_f\ (3) \\
\hline
a_m\ (5) & (9, 9) & (17, 2) \\
a_f\ (3) & (2, 17) & (8, 8)
\end{array}
\]

State transitions depend only on the current state \( s \) and all players' actions $\bm a$ or the aggregate action \( a \).Let $P(s'\mid s,a)$ denote the transition probability to state $s'$ from $s$ given aggregate action $a$:
\[
\begin{aligned}
&P(s_h \mid s_h, 6) = 1,\; P(s_h \mid s_h, 8) = \tfrac{2}{3},\; P(s_h \mid s_h, 10) = \tfrac{1}{5}, \\
&P(s_h \mid s_\ell, 6) = 1,\; P(s_h \mid s_\ell, 8) = \tfrac{1}{2},\; P(s_h \mid s_\ell, 10) = 0,
\end{aligned}
\]
with $P(s_\ell \mid s,a) = 1 - P(s_h \mid s,a)$. Consistent with the original design, a higher aggregate fishing effort $a$ increases the likelihood of transitioning to the low fish stock state $s_\ell$.
\begin{figure}[tbp]
    \centering
    \includegraphics[width=1.0\linewidth]{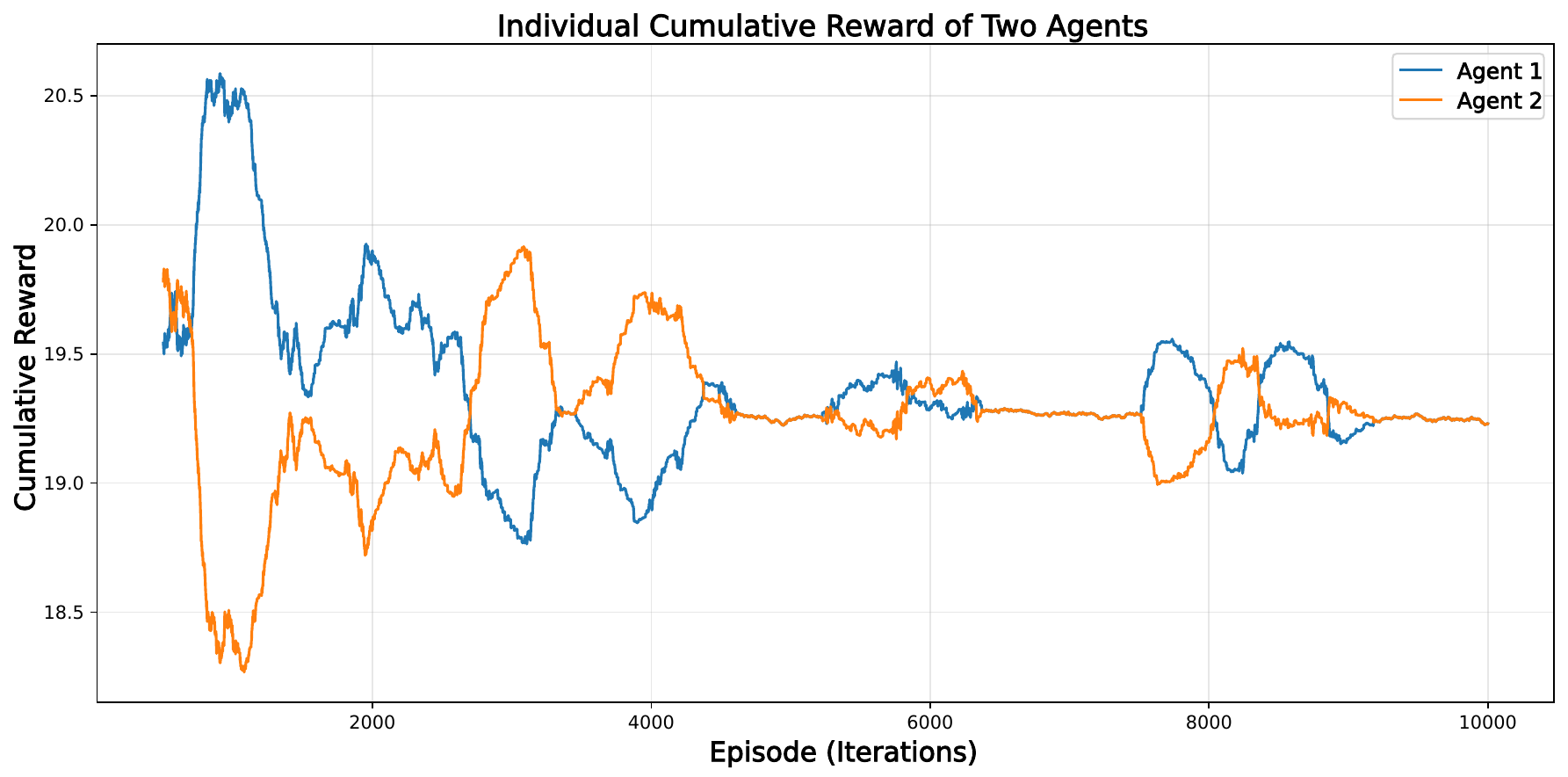}
    \caption{Individual cumulative reward of two agents.}
    \label{fig:agent_rewards}
\end{figure}

\begin{figure}[tbp]
    \centering
    \includegraphics[width=1.0\linewidth]{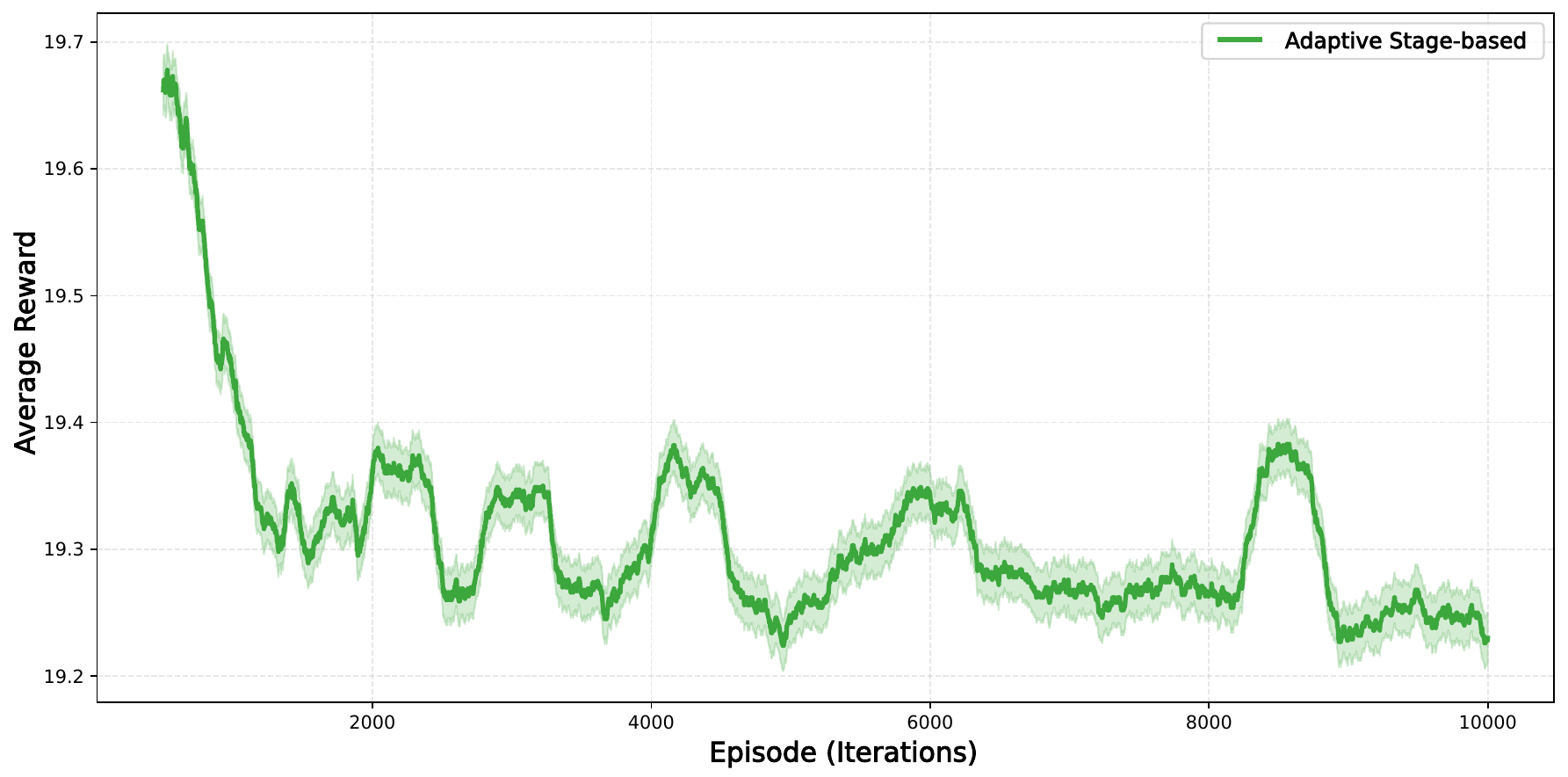}
    \caption{Average reward of two agents.}
    \label{fig:cce_convergence}
\end{figure}

\begin{figure}[tbp]
    \centering
    \includegraphics[width=1.0\linewidth]{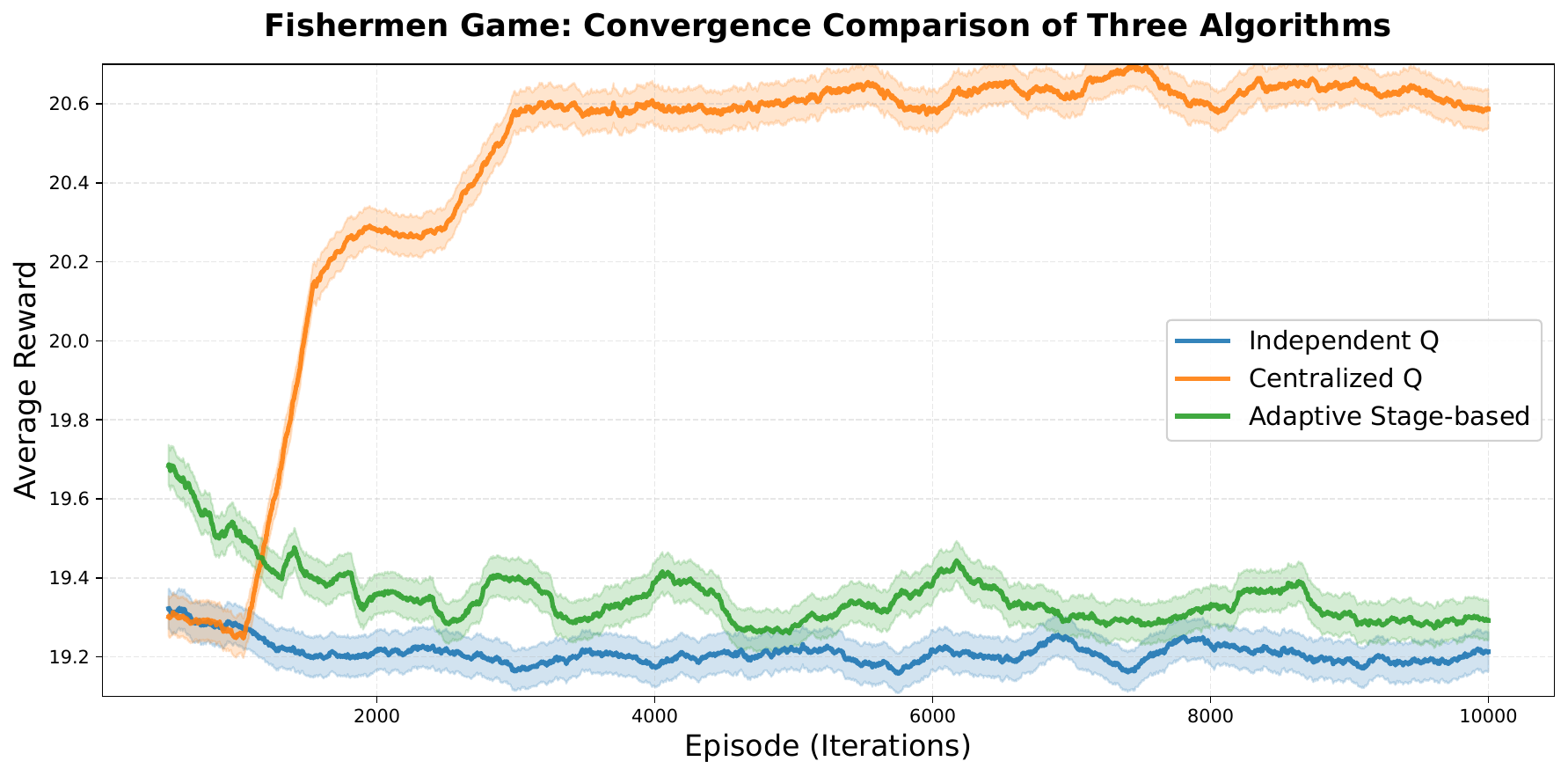}
    \caption{Rewards comparison of three algorithms on the Fishing Game. ``Centralized Q'' denotes a centralized oracle that controls all agents' actions to maximize the joint reward of both agents.
``Independent Q'' means each agent runs a naive single-agent Q-learning algorithm independently, taking greedy actions based only on local information without considering other agents. }
    \label{fig:three_algos}
\end{figure}

Figures~\ref{fig:agent_rewards} and \ref{fig:cce_convergence} illustrate the individual cumulative rewards of the two agents and the average reward, respectively. The results show that the actual policy trajectories generated by our algorithm converge to stable high rewards (around $19.2$–$19.5$ per agent after sufficient episodes), which is a desirable empirical behavior beyond our theoretical guarantee (our theoretical results only guarantee convergence for the certified output policy). Moreover, Figure~\ref{fig:agent_rewards} clearly exhibits the stage-wise behavior of our algorithm: rewards fluctuate 
noticeably at the beginning of each stage and then stabilize as learning proceeds within the stage.

Finally, we compare our algorithm with two representative baselines in Figure~\ref{fig:three_algos}: centralized Q-learning (an ideal upper bound achieved by a fully coordinated oracle that can control the actions of both agents) and naive independent Q-learning (where each agent ignores other agents and the game structure). Empirical results show that our algorithm outperforms the independent baseline, demonstrating the effectiveness of our adaptive stage-based design in the decentralized setting.

\section{Conclusion}\label{sec:conclusions}
In this paper, we studied decentralized learning of Coarse Correlated Equilibrium (CCE) in aggregative Markov games (AMGs), a class of multi-agent systems where each agent's reward depends only on its own action and an aggregate of others' actions. We proposed an adaptive stage-based V-learning algorithm that explicitly exploits the aggregative structure to enable efficient learning in a fully decentralized setting. 
We established a sample complexity guarantee showing that the proposed method learns an $\epsilon$-approximate CCE within $\widetilde{O}(S A_{\max}T^5 /\epsilon^2)$ episodes, while notably avoiding the curse of multiagents.  These results demonstrate that the aggregative structure can be harnessed to design efficient decentralized learning algorithms.

Several interesting directions remain for future research. First, it would be valuable to further tighten the sample complexity bounds, potentially by refining the analysis or leveraging sharper concentration techniques. Second, extending the proposed framework to more general classes of structured Markov games beyond the aggregative setting is a natural next step. Finally, it would be of interest to investigate whether similar ideas can be applied to learning other equilibrium concepts, such as correlated equilibrium or Nash equilibrium, in decentralized environments.

\bibliographystyle{IEEEtran}
\bibliography{refs}
\appendices

\section{Proofs for Section~\ref{sec:cce_ce}}\label{app:cce}

We first introduce a few notations to facilitate the analysis. For a step $t\in \mathcal{T}$ of an episode $k\in[K]$, we denote by $s^{t,k}$ the state that the agents observe at this time step. For any state $s^t\in\mathcal{S}^t$, we let $\pi_{i}^{t,k}(\cdot \mid s^t)\in\Delta(\mathcal{A}_i^t(s^t))$ be the distribution prescribed by Algorithm~\ref{alg:sbv} to agent $i$ at this step. Notice that such notations are well-defined for every $s^t \in\mathcal{S}^t$ even if $s^t$ might not be the state $s^{t,k}$ that is actually visited at the given step. We further let $\pi_{i}^{t,k} = \{\pi_{i}^{t,k}(\cdot\mid s^t):s^t \in\mathcal{S}^t\}$, and let $a_{i}^{t,k}\in\mathcal{A}_i^t(s^t)$ be the actual action taken by agent $i$. For any $s^t \in\mathcal{S}^t$, let $\tilde{C}^{t,k}(s^t)$ denotes the value of $\tilde{C}^{t}(s^t)$ at the \emph{beginning} of the $k$-th episode. Note that it is proper to use the same notation to denote these values from all the agents' perspectives, because the agents maintain the same estimates of these terms as they can be calculated from the common observations (of the state-visitation). We also use $\overline{V}_{t,i}^{k}(s^t)$ and $\hat{V}_{t,i}^{k}(s^t)$ to denote the values of $\overline{V}_{t,i}(s^t)$ and $\hat{V}_{t,i}(s^t)$, respectively, at the beginning of the $k$-th episode from agent $i$'s perspective.

Further, for a state $s^{t,k}$, let $\tilde{c}^{t,k}$ denote the number of times that state $s^{t,k}$ has been visited (at the $t$-th step) in the stage right before the current stage, and let $\tilde{l}_{j}^{t,k}$ denote the index of the episode that this state was visited the $j$-th time among the $\tilde{c}^{t,k}$ times. For notational convenience, we use $\tilde{c}$ to denote $\tilde{c}^{t,k}$, and $\tilde{l}_j$ to denote $\tilde{l}_{j}^{t,k}$, whenever $t$ and $k$ are clear from the context. With the new notations, the update rule in Line~\ref{line:14} of Algorithm~\ref{alg:sbv} can be equivalently expressed as
\begin{equation}\label{eqn:update_rule_new}
\hat{V}_{t,i}(s^t) \gets \frac{1}{\tilde{c}}\sum_{j=1}^{\tilde{c}} \left(r_{i}^t(s^t,\bm{a}^{t,\tilde{l}_j}) + \overline{V}_{t+1,i}^{\tilde{l}_j}(s^{t+1,\tilde{l}_j}) \right)+ b_{\tilde{c}}. 
\end{equation}

To streamline the Bellman equations for the value and Q-functions, we introduce two auxiliary operators. First, for any stage-$t$ state-action pair $(s^t,\bm{a}^t)$ and any value function $V$ defined on the stage-$t+1$ state space $\mathcal{S}^{t+1}$, the state transition operator is
$$
\mathbb{P}_{t} V(s^t, \bm{a}^t) := \mathbb{E}_{s^{t+1} \sim p^{t}(\cdot \mid s^t, \bm{a}^t)} V\left(s^{t+1}\right).
$$
Second, for a joint policy profile $\bm\pi = (\pi_i^t)_{i\in\mathcal{N},t\in\mathcal{T}}$, we define the stage-$t$ marginal joint decision rule as $\boldsymbol{\pi}^t := (\pi_1^t, \pi_2^t, \dots, \pi_N^t)$, which is the collection of stage-$t$ decision rules of all agents extracted from the full joint Markov policy $\bm\pi$. For any stage-$t$ state $s^t$ and any Q-function $Q$ defined on $(s^t,\mathcal{A}^t(s^t))$, the policy expectation operator is
$$
\mathbb{D}_{\boldsymbol{\pi}^t} Q(s^t) := \mathbb{E}_{\bm{a}^t \sim \boldsymbol{\pi}^t(\cdot \mid s^t)} Q(s^t, \bm{a}^t),
$$
where $\bm{a}^t \sim \boldsymbol{\pi}^t(\cdot \mid s^t)$ means the joint action $\bm{a}^t$ is sampled from the product distribution induced by $\boldsymbol{\pi}^t$ at state $s^t$.

With these notations and definitions, the Bellman equations for the joint value and Q-functions can be written succinctly as
\begin{align}
Q_{t}^{\bm\pi}(s^t, \bm{a}^t) &= \left(r_{t} + \mathbb{P}_{t} V_{t+1}^{\bm\pi}\right)(s^t, \bm{a}^t), \label{eqn:bellman-Q} \\
V_{t}^{\bm\pi}(s^t) &= \left(\mathbb{D}_{\boldsymbol{\pi}^t} Q_{t}^{\bm\pi}\right)(s^t) \label{eqn:bellman-V}
\end{align}
for all $t \in \mathcal{T}$, $s^t \in \mathcal{S}^t$, and $\bm{a}^t\in \mathcal{A}^t(s^t)$.

In the following proof, we assume without loss of generality that the initial state $s^1$ is fixed (i.e., the initial state distribution $\rho$ is a point mass at $s^1$). Our proof can be straightforwardly generalized to the case where $s^1$ is drawn from an arbitrary fixed distribution $\rho\in\Delta(\mathcal{S}^1)$.

In the following, we start with an intermediate result, which justifies our choice of the bonus term. 

\begin{lemma}\label{lemma:bonus}
	With probability at least $1-\frac{p}{2}$, it holds for all agents $i \in \mathcal{N},$ episodes $k \in [K],$  steps $t \in \mathcal{T},$  and states $ s^t \in \mathcal{S}^t$ that
\begin{align}
	&\max_{\pi_{i}^t} \frac{1}{\tilde{c}} \sum_{j=1}^{\tilde{c}} \mathbb{D}_{\pi_{i}^t\times \pi_{-i}^{t,\tilde{l}_j}}\left(r_{i}^t + \mathbb{P}_t\overline{V}_{t+1,i}^{\tilde{l}_j}\right)(s^t)\notag \\ &- \frac{1}{\tilde{c}} \sum_{j=1}^{\tilde{c}} \left(r_{i}^t(s^t,\bm{a}^{t,\tilde{l}_j}) + \overline{V}_{t+1,i}^{\tilde{l}_j}(s^{t+1,\tilde{l}_j}) \right) \leq 4\sqrt{T^2 A_i\iota/\tilde{c}}, \notag
    \end{align}
	 where $\iota = \log(4NS\amax KT/p)$.
\end{lemma}
\emph{Proof:}
We proceed the proof in three key steps.

\textbf{Step 1: Sample average error bound via martingale differences.}
	For a fixed agent $i\in \mathcal{N}$, episode$ k \in [K],$  step $t \in \mathcal{T},$  and state $ s^t \in \mathcal{S}^t$, let $\mathcal{F}_j$ be the $\sigma$-algebra generated by all the random variables up to episode $\tilde{l}_j$. Then, $\left\{ r_{i}^t(s^t,\bm{a}_{i}^{t,\tilde{l}_j}) + \overline{V}_{t+1,i}^{\tilde{l}_j} (s^{t+1,\tilde{l}_j}) - \mathbb{D}_{\bm{\pi}^{t,\tilde{l}_j}}\left(r_{i}^t + \mathbb{P}_t\overline{V}_{t+1,i}^{\tilde{l}_j}\right) (s^t) \right\}_{j=1}^{\tilde{c}}$ is a martingale difference sequence with respect to $\{\mathcal{F}_j\}_{j=1}^{\tilde{c}}$. By the boundedness of rewards and value functions, this sequence is bounded by $H$. Applying the Azuma-Hoeffding inequality, with probability at least $1-p/(4NSTK)$, we have
\begin{align}\label{eqn:martingale difference}
	&\frac{1}{\tilde{c}} \sum_{j=1}^{\tilde{c}} \mathbb{D}_{\bm{\pi}^{t,\tilde{l}_j}}\left(r_{i}^t + \mathbb{P}_t\overline{V}_{t+1,i}^{\tilde{l}_j}\right)(s^t) \notag \\ &\quad\quad- \frac{1}{\tilde{c}} \sum_{j=1}^{\tilde{c}} \left(r_{i}^t(s^t,\bm{a}^{t,\tilde{l}_j}) + \overline{V}_{t+1,i}^{\tilde{l}_j}(s^{t+1,\tilde{l}_j}) \right) \notag \\ &\quad\quad \leq  \sqrt{2T^2\log(8NSKT/p)/\tilde{c}} \notag \\ &\quad\quad \leq \sqrt{2T^2\iota/\tilde{c}}  \leq  \sqrt{T^2 A_i\iota/\tilde{c}}. 
\end{align}

\textbf{Step 2: Regret bound via Tsallis-INF.}
	Then, we only need to bound 
	\begin{align}\label{eqn:bandit_regret}
	R_{\tilde{c}}^\star :=& \max_{\pi_{i}^t} \frac{1}{\tilde{c}} \sum_{j=1}^{\tilde{c}} \mathbb{D}_{\pi_{i}^t\times \pi_{-i}^{t,\tilde{l}_j}}\left(r_{i}^t + \mathbb{P}_t\overline{V}_{t+1,i}^{\tilde{l}_j}\right)(s^t)\notag \\& - \frac{1}{\tilde{c}} \sum_{j=1}^{\tilde{c}} \mathbb{D}_{\bm{\pi}^{t,\tilde{l}_j}}\left( r_{i}^t + \mathbb{P}_t\overline{V}_{t+1,i}^{\tilde{l}_j}\right)(s^t).
	\end{align}
	Notice that $R_{\tilde{c}}^\star$ can be considered as the averaged regret of visiting the state $s^t$ with respect to the optimal policy in hindsight. Such a regret minimization problem can be handled by an adversarial multi-armed bandit problem, where the loss function at step $j\in[\tilde{c}]$ is defined as
	\begin{align*}
\ell_j(a_i^t) &= \mathbb{E}_{{a}_{-i}^t\sim \pi_{-i}^{t,\tilde{l}_j}(s^t)} 
\left[ T-t+1 - r_{i}^t(s^t,\bm{a}^t)  \right. \\
&\quad \left. -\mathbb{P}_t\overline{V}_{t+1,i}^{\tilde{l}_j}(s^t,\bm{a}^t) \right]/T.
\end{align*}
	Algorithm~\ref{alg:sbv} adopts the Tsallis-INF algorithm \cite{zimmert2021tsallis}, which guarantees for all $k \in [K]$:
\begin{align}\label{eqn:no-regret}
	R_{\tilde{c}}^\star &= \max_{a_{i}^t}\frac{T}{\tilde{c}}\sum_{j=1}^{\tilde{c}}(-\ell_j(a_i))- \frac{T}{\tilde{c}}\sum_{j=1}^{\tilde{c}} \mathbb{E}_{{a}_{i}^t\sim \pi_{i}^{t,\tilde{l}_j}}(s^t) [-\ell_j(a_i^t)] \notag \\&=\frac{T}{\tilde{c}}\big[\sum_{j=1}^{\tilde{c}} \mathbb{E}_{{a}_{i}^t\sim \pi_{i}^{t,\tilde{l}_j}}(s^t) [\ell_j(a_i^t)]-\min_{a_{i}^t}\sum_{j=1}^{\tilde{c}} \ell_j(a_i^t)\big] \notag \\ 
    &\leq \frac{T}{\tilde{c}} (4\sqrt{A_i \tilde{c}}+1) =4\sqrt{\frac{ T^2 A_i }{\tilde{c}}} + \frac{T}{\tilde{c}} \notag \\ 
    &\leq  3\sqrt{T^2 A_i\iota/\tilde{c}} .  
\end{align}

\textbf{Step 3: Union bound.}
Combining \eqref{eqn:martingale difference} and
\eqref{eqn:no-regret}, and taking a union bound over all agents
$i \in \mathcal{N}$, episodes $k \in [K]$, steps $t \in \mathcal{T}$, and states
$s^t \in \mathcal{S}^t$, completes the proof.
\hfill \ensuremath{\blacksquare}

\begin{algorithm}[!t]
\caption{Construction of the Correlated Policy $\bar{\bm{\pi}}^{t,k}$}\label{alg:certifyhk}
\begin{algorithmic}[1]
    \STATE \textbf{Input:} The distribution trajectory $\{\pi_{i}^{t,k}:i\in\mathcal{N},t\in \mathcal{T},k\in[K]\}$ specified by Algorithm~\ref{alg:sbv}.
    
    \STATE \textbf{Initialize:} $k'\gets k$.
    
    \FOR{step $t'\gets t$ to $T$}
        \STATE Receive $s^{t'}$;
        \STATE Uniformly sample $j$ from $\{1,2,\dots, \tilde{c}^{t', k'}(s^{t'})\}$;
        \STATE Set $k' \gets \tilde{l}_{j}^{t',k'}$; \COMMENT{$\tilde{l}_{j}^{t', k'}$ is the index of the episode that this state was visited the $j$-th time (among the total $\tilde{c}^{t',k'}$ times) in the last stage}
        \STATE Take joint action $\bm{a}^{t'}\sim \times_{i=1}^N\pi_{i}^{t', k'}(\cdot \mid s^{t'})$;
    \ENDFOR
\end{algorithmic}
\end{algorithm}

Based on the trajectory of the distributions $\{\pi_{i}^{t,k}:i\in\mathcal{N},t\in \mathcal{T},k\in[K]\}$ specified by Algorithm~\ref{alg:sbv}, we construct a correlated policy $\bar{\bm\pi}^{t,k}$ for each $(t,k)\in \mathcal{T}\times [K]$. Our construction of the correlated policies, largely inspired by the ``certified policies'' \cite{bai2020near} for learning in two-player zero-sum games, is formally presented  in Algorithm~\ref{alg:certifyhk}. We further define an output policy $\bar{\bm\pi}$: first uniformly sample an index $k$ from $[K]$, then execute the  policy $\bar{\bm\pi}^{1,k}$. A formal description of $\bar{\bm\pi}$ is given in Algorithm~\ref{alg:certify}. By construction of the correlated policies $\bar{\bm\pi}^{t,k}$, we know that for any $i \in \mathcal{N},$ episode $ k \in [K],$  step $t \in [T+1],$  and state $ s^t \in \mathcal{S}^t$, the corresponding value function can be written recursively as follows:
\[
V_{t,i}^{\bar{\bm\pi}^{t,k}}(s^t) = \frac{1}{\tilde{c}}\sum_{j=1}^{\tilde{c}} \mathbb{D}_{\bm\pi^{t,\tilde{l}_j}}\bigg(r_{i}^t+\mathbb{P}_{t} V_{t+1,i}^{\bar{\bm\pi}^{t+1,\tilde{l}_j}}\bigg)(s^t), 
\]
and $V_{t,i}^{\bar{\bm\pi}^{t,k}}(s^t) = 0$ if $t=T+1$ or $k$ is in the first stage of the corresponding $(t,s^t)$ pair. We also immediately obtain that
\[
V_{1,i}^{\bar{\bm\pi}}(s^1) = \frac{1}{K}\sum_{k=1}^K V_{1,i}^{\bar{\bm\pi}^{1,k}}(s^1).
\]

For analytical purposes, we introduce two auxiliary notations $\underline{V}$ and  $\underline{\hat{V}}$ as lower confidence bounds of the value estimates. Specifically, for any $i \in \mathcal{N},$ episode $ k \in [K],$  step $t \in [T+1],$  and state $ s^t \in \mathcal{S}^t$, we define $\underline{V}_{t,i}^k(s^t) = \underline{\hat{V}}_{t,i}^k(s^t) = 0$ if $t = T+1$ or $k$ is in the first stage of the $(h,s^t)$ pair, and 
\begin{align*}
\underline{\hat{V}}_{t,i}^k(s^t) &=\frac{1}{\tilde{c}}\sum_{j=1}^{\tilde{c}} \left(r_{i}^t(s^t,\bm{a}^{t,\tilde{l}_j}) + \underline{V}_{t+1,i}^{\tilde{l}_j}(s^{t+1,\tilde{l}_j}) \right)- b_{\tilde{c}},\\ &\text{ and } \underline{V}_{t,i}^k(s^t) = \max\left\{\underline{\hat{V}}_{t,i}^k(s^t), 0\right\}. 
\end{align*}
Note that these notations are only for analysis and agents do not need to maintain them explicitly during learning. Further, recall that $V_{t,i}^{\star,\bar{\pi}_{-i}^{t,k}}(s^t)$ is agent $i$'s best response value against its opponents' policy $\bar{\pi}_{-i}^{t,k}$. Our next lemma shows that $\overline{V}_{t,i}^k(s^t)$ and $\underline{V}_{t,i}^k(s^t)$ are indeed valid upper and lower bounds of $V_{t,i}^{\star,\bar{\pi}_{-i}^{t,k} }(s^t)$ and $V_{t,i}^{\bar{\bm\pi}^{t,k}}(s^t)$, respectively. 

\begin{lemma}\label{lemma:upperbound}
	It holds with probability at least $1-p$ that for all  for all agents $i \in \mathcal{N},$ episodes $k \in [K],$  steps $t \in \mathcal{T},$  and states $ s^t \in \mathcal{S}^t$ ,
\begin{align}
	 V_{t,i}^{\star,\bar{\pi}_{-i}^{t,k}}(s^t) \leq \overline{V}_{t,i}^k(s^t), \label{first inequality}\\
      V_{t,i}^{\bar{\bm\pi}^{t,k}}(s^t) \geq \underline{V}_{t,i}^k(s^t) . \label{second inequality}
\end{align}
\end{lemma}
\emph{Proof:}
	Fix an agent $i\in \mathcal{N}$, episode$ k \in [K],$  step $t \in \mathcal{T},$  and state $ s^t \in \mathcal{S}^t$. The desired result clearly holds for any state $s^t$ that is in its first stage, due to our initialization of $\overline{V}_{t,i}^k(s^t)$ and  $\underline{V}_{t,i}^k(s^t)$ for this special case. In the following, we only need to focus on the case where $\overline{V}^k_{t,i}(s^t)$ and  $\underline{V}_{t,i}^k(s^t)$ have been updated at least once at the given state $s^t$ before the $k$-th episode. 

	We first prove the first inequality \eqref{first inequality}. It suffices to show that $\hat{V}_{t,i}^k(s^t)\geq V_{t,i}^{\star,\bar{\pi}_{-i}^{t,k}}(s^t)$, since $\overline{V}_{t,i}^k(s^t)=\min\{\hat{V}_{t,i}^k(s^t), T-t+1\}$, and $V_{t,i}^{\star,\bar{\pi}_{-i}^{t,k}}(s^t)$ is always less than or equal to $T-t+1$. Our proof relies on induction on $k\in [K]$. The base case $k=1$ holds by initialization logic. For the inductive step, consider two cases for $t \in \mathcal{T}$ and $s^t \in \mathcal{S}^t$:
	
	\textbf{Case 1 (for \eqref{first inequality}):} $\hat{V}_{t,i}(s^t)$ has just been updated in (the end of) episode $k-1$. In this case, by definition of stage-based updates:
	\begin{equation}
	\hat{V}^k_{t,i}(s^t) = \frac{1}{\tilde{c}}\sum_{j=1}^{\tilde{c}} \left( r_{i}^t(s^t,\bm{a}^{t,\tilde{l}_j}) + \overline{V}_{t+1,i}^{\tilde{l}_j}(s^{t+1,\tilde{l}_j}) \right) + b_{\tilde{c}} .
	\end{equation}
	And by the definition of $V_t^{\star, \bar{\pi}_{-i}^{t,k}}(s^t)$, it holds with probability at least $1-\frac{p}{4NSKT}$ that
	\begin{align}
	V_{t,i}^{\star,\bar{\pi}_{-i}^{t,k}}(s^t)&\leq  \max_{\pi_{i}^t} \frac{1}{\tilde{c}}\sum_{j=1}^{\tilde{c}} \mathbb{D}_{\pi_{i}^t\times \pi_{-i}^{t,\tilde{l}_{j}}} \left( r_{i}^t + \mathbb{P}_tV_{t+1,i}^{\star, \bar{\pi}_{-i}^{t+1,\tilde{l}_{j}}}\right)  (s^t)\nonumber\\
	&\leq\max_{\pi_{i}^t} \frac{1}{\tilde{c}}\sum_{j=1}^{\tilde{c}} \mathbb{D}_{\pi_{i}^t\times \pi_{-i}^{t,\tilde{l}_{j}}} \left( r_{i}^t + \mathbb{P}_t\overline{V}_{t+1,i}^{\tilde{l}_j}\right)  (s^t)\nonumber\\
	\leq  \frac{1}{\tilde{c}} \sum_{j=1}^{\tilde{c}}& \left( r_{i}^t(s^t,\bm{a}^{t,\tilde{l}_j}) + \overline{V}_{t+1,i}^{\tilde{l}_j}(s^{t+1,\tilde{l}_j}) \right) + 4\sqrt{T^2 A_i\iota/\tilde{c}}\nonumber\\
	&\leq  \hat{V}^k_{t,i}(s^t),
	\end{align}
	where the second step is by the induction hypothesis, the third step holds due to Lemma~\ref{lemma:bonus}, and the last step is by the definition of $b_{\tilde{c}}$. 
	
	\textbf{Case 2 (for \eqref{first inequality}): } $\hat{V}_{t,i}(s^t)$ was not updated in (the end of) episode $k-1$.  Since we have excluded the case that $\hat{V}_{t,i}$ has never been updated, we are guaranteed that there exists an episode $j$ such that $\hat{V}_{t,i}(s^t)$ has been updated in the end of episode $j-1$ most recently. In this case, $\hat{V}_{t,i}^k(s^t) = \hat{V}_{t,i}^{k-1}(s^t) = \dots = \hat{V}_{t,i}^{j}(s^t) \geq V_{t,i}^{\star,\bar{\pi}_{-i}^{t,j}}(s^t)$, where the last step is by the induction hypothesis. Finally, by the construction of stage-based policies (Algorithm~\ref{alg:sbv}), $\bar{\pi}_{-i}^{t,j}$ is constant within the same stage (i.e., unchanged for episodes in that stage), so $V_{t,i}^{\star,\bar{\pi}_{-i}^{t,j}}(s^t)$ is also constant for all episodes $j$ in the same stage. Since we know that episode $j$ and episode $k$ lie in the same stage, we can conclude that $V_{t,i}^{\star,\bar{\pi}_{-i}^{t,k}}(s^t) = V_{t,i}^{\star,\bar{\pi}_{-i}^{t,j}}(s^t) \leq \hat{V}_{t,i}^k(s^t)$. 
	
	Combining the two cases and applying a union bound over all agents $i \in \mathcal{N},$ episodes $k \in [K],$  steps $t \in \mathcal{T},$  and states $ s^t \in \mathcal{S}^t$, the first inequality holds with probability at least $1-\frac{p}{2}$. 
	
	Next, we prove the second inequality \eqref{second inequality} in the statement of the lemma. Notice that it suffices to show $\underline{\hat{V}}_{t,i}^k(s^t)\leq V_{t,i}^{\bar{\bm\pi}^{t,k}}(s^t)$ because $\underline{V}_{t,i}^k(s^t) = \max\{\underline{\hat{V}}_{t,i}^k(s^t),0\}$. Our proof again relies on induction on $k\in[K]$. Similar to the proof of the first inequality, the claim apparently holds for $k=1$, and we consider the following two cases for each step $t\in \mathcal{T}$ and $s^t \in\mathcal{S}^t$.
	
	\textbf{Case 1 (for \eqref{second inequality}):} The value of $\underline{\hat{V}}_{t,i}(s^t)$ has just changed in (the end of) episode $k-1$. In this case, 
	\begin{equation}
	\underline{\hat{V}}^k_{t,i}(s^t) = \frac{1}{\tilde{c}}\sum_{j=1}^{\tilde{c}} \left( r_{i}^t(s^t,\bm{a}^{t,\tilde{l}_j}) + \underline{\hat{V}}_{t+1,i}^{\tilde{l}_j}(s^{t+1,\tilde{l}_j}) \right)- b_{\tilde{c}} .
	\end{equation}
	By the definition of $V_{t,i}^{\bar{\bm\pi}^{t,k}}(s^t)$, it holds with probability at least $1-\frac{p}{4NSKT}$ that
	\begin{align}
	V_{t,i}^{\bar{\bm\pi}^{t,k}}(s^t)=  &\frac{1}{\tilde{c}}\sum_{j=1}^{\tilde{c}} \mathbb{D}_{\bm{\pi}^{t,\tilde{l}_j}} \left( r_{i}^t+\mathbb{P}_{t} V_{t+1,i}^{\bar{\bm\pi}^{t+1, \tilde{l}_j}}\right)(s^t)\nonumber\\
	\geq&\frac{1}{\tilde{c}}\sum_{j=1}^{\tilde{c}} \mathbb{D}_{\bm{\pi}^{t,\tilde{l}_j}} \left( r_{i}^t+\mathbb{P}_{t} \underline{\hat{V}}_{t+1,i}^{\tilde{l}_j}\right)(s^t)\nonumber\\
	\geq  \frac{1}{\tilde{c}} \sum_{j=1}^{\tilde{c}} & \left( r_{i}^t(s^t,\bm{a}^{t,\tilde{l}_j}) +\underline{\hat{V}}_{t+1,i}^{\tilde{l}_j}(s^{t+1,\tilde{l}_j}) \right) - \sqrt{2T^2 \iota/\tilde{c}}\nonumber\\
	\geq & \underline{\hat{V}}^k_{t,i}(s^t),
	\end{align}
	where the second step is by the induction hypothesis, the third step holds due to the Azuma-Hoeffding inequality, and the last step is by the definition of $b_{\tilde{c}}$. 
	
	\textbf{Case 2 (for \eqref{second inequality}): } The value of $\underline{\hat{V}}_{t,i}(s^t)$ has not changed in (the end of) episode $k-1$.  Since we have excluded the case that $\underline{\hat{V}}_{t,i}$ has never been updated, we are guaranteed that there exists an episode $j$ such that $\underline{\hat{V}}_{t,i}(s^t)$ has changed in the end of episode $j-1$ most recently. In this case, we know that indices $j$ and $k$ belong to the same stage, and $\underline{\hat{V}}_{t,i}^k(s^t) = \underline{\hat{V}}_{t,i}^{k-1}(s^t) = \dots = \underline{\hat{V}}_{t,i}^{j}(s^t) \leq V_{t,i}^{\bar{\bm\pi}^{t,j}}(s^t)$, where the last step is by the induction hypothesis. Finally, by stage-based policy construction (Algorithm~\ref{alg:sbv}), $\bar{\bm\pi}^{t,j}$ is constant within the same stage, so $V_{t,i}^{\bar{\bm\pi}^{t,j}}(s^t)$ is constant for all episodes $j$ in that stage. Since $j$ and $k$ lie in the same stage, we can conclude  that $V_{t,i}^{\bar{\bm\pi}^{t,k}}(s^t)= V_{t,i}^{\bar{\bm\pi}^{t,j}}(s^t) \geq \underline{\hat{V}}_{t,i}^k(s^t)$. 
	
	Again, combining the two cases and applying a union bound over all agents
$i \in \mathcal{N}$, episodes $k \in [K]$, steps $t \in \mathcal{T}$, and states
$s^t \in \mathcal{S}^t$ , the second inequality holds with probability at least $1-\frac{p}{2}$.

By the union bound over both inequalities, the lemma holds with probability at least $1-p$. 
\hfill \ensuremath{\blacksquare}

The following result shows that the agents have no incentive to deviate from the correlated policy $\bar{\bm\pi}$, up to a regret term of the order $\widetilde{O}(\sqrt{T^5 S \amax /K})$. 

\begin{theorem}\label{thm:main}
	For any $p\in(0,1]$, let $\iota = \log(2NS\amax KT/p)$. Suppose $K \geq \frac{ST}{\amax\iota}$, with probability at least $1-p$, the following holds for any initial state $s^1\in \mathcal{S}^1$ and agent $i \in \mathcal{N}$:
	$$
	V_{1,i}^{\star,\bar{\bm\pi}_{-i}}(s^1) - V_{1,i}^{\bar{\bm\pi}}(s^1)\leq O \left(\sqrt{T^5 SA_{\max} \iota/K}\right).
	$$   
\end{theorem}
\emph{Proof:}
	We first recall the definitions of several notations and define a few new ones. For a state $s^{t,k}$, recall that $\tilde{c}^{t,k}$ denotes the number of visits to the state $s^{t,k}$ (at the $t$-th step)  in the stage right before the current stage, and $\tilde{l}_{j}^{t,k}$ denotes the $j$-th episode among the $\tilde{c}^{t,k}$ episodes. When $t$ and $k$ are clear from context, we abbreviate $\tilde{l}_{j}^{t,k}$ as $\tilde{l}_j$ and $\tilde{c}^{t,k}$ as $\tilde{c}$.
    
    By Lemma~\ref{lemma:upperbound} (upper/lower bound properties of $\overline{V}$ and $\underline{V}$) and the construction of the output policy $\bar{\bm\pi}$ in  Algorithm~\ref{alg:sbv}, we know that
	\[
	\begin{aligned}
	V_{1,i}^{\star,\bar{\bm\pi}_{-i}}(s^1) - V_{1,i}^{\bar{\bm\pi}}(s^1) \leq & \frac{1}{K}\sum_{k=1}^K  \left( V_{1,i}^{\star,\bar{\bm\pi}_{1,-i}^k}(s^1) - V_{1,i}^{\bar{\bm\pi}_1^k} (s^1) \right)\\
	\leq & \frac{1}{K}\sum_{k=1}^K  \left( \overline{V}_{1,i}^{k}(s^1) - \underline{V}_{1,i}^{k}(s^1)\right). 
	\end{aligned}
	\]
	Thus, it suffices to upper bound $\frac{1}{K}\sum_{k=1}^K ( \overline{V}_{1,i}^{k}(s^1) - \underline{V}_{1,i}^{k}(s^1))$. For a fixed agent $i\in\mathcal{N}$, we define the following notation:
	$$
	\delta^{t,k} := \overline{V}_{t,i}^{k}(s^{t,k})- \underline{V}_{t,i}^{k}(s^{t,k}).
	$$ 
	The key idea of the subsequent proof is to upper bound $\sum_{k=1}^K \delta^{t,k}$ by the next step $\sum_{k=1}^K \delta^{t+1,k}$, and then obtain a recursive formula. From the update rule of $\overline{V}_{t,i}^k(s^{t,k})$ in \eqref{eqn:update_rule_new}, we have:
    \begin{align*}
	\overline{V}_{t,i}^k(s^{t,k}) \leq & \mathbb{I}[c^{t,k} = 0]T  \\ & + \frac{1}{\tilde{c}}\sum_{j=1}^{\tilde{c}}  \left( r_{i}^t(s^t,\bm{a}^{t,\tilde{l}_j}) + \overline{V}_{t+1,i}^{\tilde{l}_j}(s^{t+1,\tilde{l}_j}) \right)+ b_{\tilde{c}},
	\end{align*}
	where the $\mathbb{I}[c^{t,k}=0]$ term counts for the event that the optimistic value function has never been updated for the given state.
	
	Combining this with the definition of $\underline{V}_{t,i}^{k}(s^{t,k})$, we have
	\begin{align}
	\delta^{t,k} \leq & \mathbb{I}[c^{t,k} = 0]T \notag \\
    &+ \frac{1}{\tilde{c}}\sum_{j=1}^{\tilde{c}}  \left( \overline{V}_{t+1,i}^{\tilde{l}_j}(s^{t+1,\tilde{l}_j}) - \underline{V}_{t+1,i}^{\tilde{l}_j}(s^{t+1,\tilde{l}_j}) \right)+2 b_{\tilde{c}}\nonumber\\
	\leq& \mathbb{I}[c^{t,k} = 0]T + \frac{1}{\tilde{c}}\sum_{j=1}^{\tilde{c}}\delta^{t+1,\tilde{l}_j} + 2b_{\tilde{c}},\label{eqn:delta}
	\end{align}
	To find an upper bound of $\sum_{k=1}^K \delta^{t,k}$, we proceed to upper bound each term on the RHS of \eqref{eqn:delta} separately:
    
\textbf{Term 1:} $\sum_{k=1}^K \mathbb{I}[c^{t,k}=0]T$.
Each fixed step-state pair $(t,s^t)$ contributes at most $1$ to $\sum_{k=1}^K \mathbb{I}\left[c^{t,k}=0\right]$. There are $ST$ such pairs (across $t\in\mathcal{T}$ and $s^t\in\mathcal{S}^t$), so $\sum_{k=1}^K \mathbb{I}[c^{t,k}=0]T \leq ST^2$.

\textbf{Term 2:} $\sum_{k=1}^K \frac{1}{\tilde{c}^{t,k}}\sum_{j=1}^{\tilde{c}^{t,k}} \delta^{t+1,\tilde{l}_{j}^{t,k}}$.
Switching the order of summation, we rewrite this term as:
	\begin{align}
	\sum_{k=1}^K \frac{1}{\tilde{c}^{t,k}}\sum_{j=1}^{\tilde{c}^{t,k}} \delta^{t+1,\tilde{l}_{j}^{t,k}} =& \sum_{k=1}^K \sum_{m=1}^K \frac{1}{\tilde{c}^{t,k}} \delta^{t+1,m} \sum_{j=1}^{\tilde{c}^{t,k}} \mathbb{I}\left[\tilde{l}_{j}^{t,k} = m\right]\nonumber\\
	=&\sum_{m=1}^K  \delta^{t+1,m} \sum_{k=1}^K  \frac{1}{\tilde{c}^{t,k}} \sum_{j=1}^{\tilde{c}^{t,k}} \mathbb{I}\left[\tilde{l}_{j}^{t,k} = m\right].\label{eqn:d1}
	\end{align}
	For a fixed episode $m$, notice that $\sum_{j=1}^{\tilde{c}^{t,k}} \mathbb{I}[\tilde{l}_{j}^{t,k} = m]\leq 1$, and that $\sum_{j=1}^{\tilde{c}^{t,k}} \mathbb{I}[\tilde{l}_{j}^{t,k} = m]= 1$ happens if and only if $s^{t,k} = s^{t,m}$ and $(t,m)$ lies in the previous stage of $(t,k)$ with respect to the step-state pair $(t,s^{t,k})$. 
    Define $\mathcal{K}_m:= \{ k\in[K]: \sum_{j=1}^{\tilde{c}^{t,k}} \mathbb{I}[\tilde{l}_{j}^{t,k} = m]= 1 \}$; this set consists of all episodes $k$ for which $m$ is a pre-stage episode of $(t,s^{t,k})$. Then we know that all episode indices $k\in \mathcal{K}_m$ belong to the same stage, and hence these episodes have the same value of $\tilde{c}^{t,k}$. That is, there  exists an integer $N_m>0$, such that $\tilde{c}^{t,k} = N_m,\forall k \in \mathcal{K}_m$. Further, since the stages are partitioned in a way such that each stage is at most $(1+\frac{1}{T})$ times longer than the previous stage (Line~\ref{line:17} of Algorithm~\ref{alg:sbv}), we know that $|\mathcal{K}_m|\leq (1+\frac{1}{T})N_m$. Therefore, for every $m$, it holds that
	\begin{equation}\label{eqn:d2}
	\sum_{k=1}^K  \frac{1}{\tilde{c}^{t,k}} \sum_{j=1}^{\tilde{c}^{t,k}} \mathbb{I}\left[\tilde{l}_{j}^{t,k} = m\right] \leq 1+ \frac{1}{T}.
	\end{equation} 
	Combining \eqref{eqn:d1} and \eqref{eqn:d2} leads to the following upper bound of the second term in \eqref{eqn:delta}:	
	\begin{equation}\label{eqn:d3}
	\sum_{k=1}^K \frac{1}{\tilde{c}^{t,k}}\sum_{j=1}^{\tilde{c}^{t,k}} \delta^{t+1,\tilde{l}_{t,j}^k}   \leq (1+\frac{1}{T}) \sum_{k=1}^{K} \delta^{t+1,k}.
	\end{equation}
	So far, we have obtained the following upper bound:
	\[
	\sum_{k=1}^K\delta^{t,k} \leq ST^2 + (1+\frac{1}{T})  \sum_{k=1}^K \delta^{t+1,k} + 2\sum_{k=1}^K b_{\tilde{c}^{t,k}}. 
	\]
	
	Iterating the above inequality over $t = T, T-1, \dots, 1$ leads to
	\begin{align}\label{eqn:tmp9}
	\sum_{k=1}^{K}\delta^{1,k} &\leq \left( \sum_{t=1}^T (1+\frac{1}{T})^{t-1} \right) ST^2+ 2\sum_{t=1}^T \sum_{k=1}^K (1+\frac{1}{T})^{t-1}b_{\tilde{c}^{t,k}} \notag \\ &=  \left((1+\frac{1}{T})^T-1 \right) ST^3 + 2\sum_{t=1}^T \sum_{k=1}^K (1+\frac{1}{T})^{t-1}b_{\tilde{c}^{t,k}}  \notag \\&\leq O\left( ST^3 + \sum_{t=1}^T \sum_{k=1}^K (1+\frac{1}{T})^{t-1}b_{\tilde{c}^{t,k}} \right),
	\end{align}
	where we used the fact that $(1+\frac{1}{T})^T \leq e$. 

    \textbf{Term 3: Analysis of $b_{\tilde{c}^{t,k}}$.} 
    In the following, we analyze the bonus term $b_{\tilde{c}^{t,k}}$.Recall that for any state $s^t$, when the number of visits $\tilde{C}^{t,k}$ to $(t,s^t)$ reaches the predefined stage length (i.e., $\tilde{C}^t(s^t) = \tilde{L}^t(s^t)$), the algorithm enters a new stage. The adaptive stage length is updated as $\tilde{L}^t(s^t)\gets \lfloor{\lambda(s^t) (1+\frac{1}{T})\tilde{L}^t(s^t)}\rfloor$, where $\lambda(s^t) \in (\frac{T}{T+1},1]$; the bonus term is defined as $b_{\tilde{c}}= 4\sqrt{T^2 A_i\iota/\tilde{c}}$. 
    
    For convenience, we introduce auxiliary notations for adaptive stage lengths: let $e_{s^t,1} = T$ and $e_{s^t,j+1} = \lfloor{\lambda(s^t)(1+\frac{1}{T})e_{s^t,j}}\rfloor$ for $j\geq 1$. For any $t\in \mathcal{T}$, we decompose the sum over $k$ by state $s^t$ and stage length $e_{s^t,j}$: 
	\begin{align}
	\sum_{k=1}^K (1+\frac{1}{T}&)^{t-1} b_{\tilde{c}^{t,k}} \leq \sum_{k=1}^K (1+\frac{1}{T})^{t-1}4\sqrt{T^2 A_i\iota/\tilde{C}^{t,k}}\nonumber\\
	=&4\sqrt{T^2 A_i\iota}\sum_{s^t \in\mathcal{S}^t}\sum_{j\geq 1} (1+\frac{1}{T})^{t-1} \times \nonumber\\ & e_{s^t,j}^{-\frac{1}{2}}\sum_{k=1}^K\mathbb{I}\left[s^{t,k} = s^t, \tilde{C}^{t,k}(s^{t,k}) = e_{s^t,j}\right]\nonumber\\
	=& 4\sqrt{T^2 A_i\iota}\sum_{s^t \in\mathcal{S}^t}\sum_{j\geq 1} (1+\frac{1}{T})^{t-1}w(s^t,j)e_{s^t,j}^{-\frac{1}{2}},\nonumber
	\end{align}
	where we define
\[
w(s^t,j):= \sum_{k=1}^K\mathbb{I}\{s^{t,k} = s^t, \tilde{C}^{t,k}(s^{t,k}) = e_j\},
\]
for any $s^t \in\mathcal{S}^t$ (number of episodes where $(t,s^t)$ has pre-stage length $e_{s^t,j}$). If we further let $w(s^t) := \sum_{j\geq 1} w(s^t,j)$, we can see that $\sum_{s^t \in\mathcal{S}^t} w(s^t) = K$. 
    For each fixed state $s^t$, we now seek an upper bound of its corresponding $j$ value, denoted as $J$ in what follows.
   By the stage update rule, for any $1\leq j\leq J$:
\begin{align*}
&\lfloor{\lambda_{\min}(1+\frac{1}{T})e_{s^t,j}}\rfloor \\ &\leq w(s^t,j)   =\sum_{k=1}^K \mathbb{I}\left[s^{t,k}=s^t, \tilde{C }^{t,k}(s^{t,k}) = e_{s^t,j}\right]  \\ &\leq \lfloor{(1+\frac{1}{T})e_{s^t,j}}\rfloor,
\end{align*}
where $\lambda_{\min}\in(\frac{T}{T+1},1]$. Thus, the sequence $\{e_{s^t,j}\}$ grows almost geometrically with ratio $\rho \in(1,1+\frac{1}{T}]$, which can also be written as $\rho = 1+\frac{c}{T}$ for some constant $c \in (0,1]$. By the formula for the sum of a geometric sequence, it follows that
\begin{align*}
\sum_{j=1}^J e_{s^t,j} = \Theta\left( \frac{T}{c} e_{s^t,1}[\rho^J-1] \right) = \Theta\left( T e_{s^t,J} \right),\\
\sum_{j=1}^J e_{s^t,j}^{\frac{1}{2}} = \Theta\left( \frac{e_{s^t,1}^{\frac{1}{2}}[\rho^{\frac{J}{2}}-1]}{\sqrt{1+\frac{1}{T}}-1} \right) = \Theta\left( T e_{s^t,J}^{\frac{1}{2}} \right).
\end{align*}
Therefore, we have 
	\[
	\begin{aligned} 
	\sum_{j\geq 1} (1+\frac{1}{T})^{t-1}w(s^t,j)e_{s^t,j}^{-\frac{1}{2}}\leq O\left(\sum_{j=1}^J e_{s^t,j}^{\frac{1}{2}}\right)\leq O\left(\sqrt{w(s^t)T}\right), 
	\end{aligned}
	\]
	 Finally, using the fact that $\sum_{s^t \in\mathcal{S}^t}w(s^t) = K$ and applying the Cauchy-Schwartz inequality, we have
	\begin{align}
	&\sum_{t=1}^T\sum_{k=1}^K (1+\frac{1}{T})^{t-1} b_{\tilde{c}^{t,k}}  \nonumber\\ & = O\left(\sqrt{T^4 A_i\iota}\sum_{s^t \in \mathcal{S}^t}\sum_{j\geq 1} (1+\frac{1}{T})^{t-1}w(s^t,j)e_j^{-\frac{1}{2}}\right)\nonumber\\
	&\leq O\left(\sqrt{T^4 A_i\iota}\sum_{s^t \in\mathcal{S}^t}\sqrt{w(s^t)T}  \right) \leq O\left(\sqrt{SA_i KT^5\iota}\right).\label{eqn:tmp8}
	\end{align}
	Summarizing the results above leads to
	\[
	\sum_{k=1}^K \delta_1^k \leq O\left( ST^3 + \sqrt{SA_i KT^5\iota} \right). 
	\]
	In the case when $K$ is large enough, such that $K \geq \frac{ST}{A_i\iota}$, the second term becomes dominant, and we obtain the desired result:	
	\[
	V_{1,i}^{\star,\bar{\bm\pi}_{-i}}(s^1) - V_{1,i}^{\bar{\bm\pi}}(s^1) \leq \frac{1}{K}\sum_{k=1}^{K}\delta_1^k \leq O\left(\sqrt{SA_i T^5\iota/K} \right) .
	\]
	This completes the proof of the theorem.
\hfill \ensuremath{\blacksquare}

An immediate corollary is that we obtain an $\epsilon$-approximate CCE when $\sqrt{S\amax T^5\iota/K}\leq \epsilon$, which is Theorem~\ref{thm:cce} in the main text.

\vspace{.8em}
\noindent\textbf{Theorem \ref{thm:cce}.} (Sample complexity of learning CCE). For any $p\in (0,1]$, set $\iota = \log(2NSA_{\max} KT/p)$, and let the agents run Algorithm~\ref{alg:sbv} for $K$ episodes with $K= O(S A_{\max}T^5 \iota/\epsilon^2)$. Then, with probability at least $1-p$, the output policy $\bar{\bm\pi}$ constitutes an $\epsilon$-approximate coarse correlated equilibrium. 

\end{document}